\documentclass[10pt]{amsart}
\usepackage{amssymb,mathrsfs,color}
\usepackage{pinlabel}

\usepackage{amsmath, amsfonts, amsthm, verbatim}
\usepackage{epstopdf, mathtools}
\mathtoolsset{showonlyrefs}
\usepackage{color}
\usepackage{esint}
\usepackage{stmaryrd}
\usepackage{caption}
\usepackage{empheq}
\usepackage[shortlabels]{enumitem}

\newcommand{\rar}{\rightarrow}

\newcommand{\bs}[1]{\boldsymbol{#1}}

\renewcommand{\div}{\operatorname{div}}

\definecolor{deepgreen}{cmyk}{1,0,1,0.5}

\newcommand{\al}{\alpha}

\newcommand{\p}{\partial}

\numberwithin{equation}{section}

\newtheorem{thm}{Theorem}[section]

\newtheorem{prop}[thm]{Proposition}

\theoremstyle{remark}

\newtheorem{assum}[thm]{Assumption}

\newtheorem{defn}[thm]{Definition}

\usepackage{tensor}

\usepackage{graphicx}
\usepackage{xcolor} 
\usepackage{tensor}
\usepackage{slashed}
\usepackage{cite}
\definecolor{green}{rgb}{0,0.8,0} 

\usepackage[bottom]{footmisc}

\newcommand{\tr}{\textrm{tr}}

\newcommand{\bbE}{\mathbb E}

\newcommand{\bbR}{\mathbb R}

\newcommand{\mfm}{\bs{\mathsf{m}}}
\newcommand{\mfn}{\bs{\mathsf{n}}}

\newcommand{\mfv}{\bs{\mathsf{v}}}
\newcommand{\mfa}{\bs{\mathsf{a}}}
\newcommand{\mfb}{\bs{\mathsf{b}}}

\newcommand{\mfu}{\bs{\mathsf{u}}}
\newcommand{\mfM}{\bs{\mathsf{M}}}
\newcommand{\mfN}{\bs{\mathsf{N}}}
\newcommand{\mfF}{\bs{\mathsf{F}}}
\newcommand{\mfG}{\bs{\mathsf{G}}}
\newcommand{\mfA}{\bs{\mathsf{A}}}
\newcommand{\mfB}{\bs{\mathsf{B}}}

\newcommand{\tbsu}{\tilde{\bs u}}
\newcommand{\tbsv}{\tilde{\bs v}}

\newcommand{\tbsR}{\tilde{\bs R}}
\newcommand{\tbsd}{\tilde{\bs d}}

\newcommand{\msm}{\bs{\mathsf{m}}}
\newcommand{\msn}{\bs{\mathsf{n}}}
\newcommand{\msv}{\bs{\mathsf{v}}}
\newcommand{\msa}{\bs{\mathsf{a}}}
\newcommand{\msb}{\bs{\mathsf{b}}}

\newcommand{\msu}{\bs{\mathsf{u}}}
\newcommand{\msM}{\bs{\mathsf{M}}}
\newcommand{\msN}{\bs{\mathsf{N}}}
\newcommand{\msX}{\bs{\mathsf{X}}}

\newcommand{\msY}{\bs{\mathsf{Y}}}

\newcommand{\cl}{\mathcal}
\newcommand{\tens}{\otimes}

\begin{document}

\title[Evolving natural configurations]{Special Cosserat rods with rate-dependent evolving natural configurations}

\author{K. R. Rajagopal and C. Rodriguez}

\begin{abstract}
We present a nonlinear, geometrically exact, and thermodynamically consistent framework for modeling special Cosserat rods with evolving natural configurations. In contrast to the common usage of the point-wise Clausius-Duhem inequality to embody the Second Law of Thermodynamics, we enforce the strictly weaker form that the rate of \emph{total} entropy production is non-decreasing. The constitutive relations between the state variables and applied forces needed to close the governing field equations are derived via prescribing frame indifferent forms of the Helmholtz energy and the total dissipation rate and requiring that the state variables evolve in a way that maximizes the rate of total entropy production. Due to the flexibility afforded by enforcing a global form of the Second Law, there are two models obtained from this procedure: one satisfying the stronger form of the Clausius-Duhem inequality and one satisfying the weaker global form of the Clausius-Duhem inequality. Finally, we show that in contrast to other viscoelastic Cosserat rod models introduced in the past, certain quadratic strain energies in our model yield \emph{both} solid-like stress relaxation and creep.     
\end{abstract} 

\maketitle

\section{Introduction}

\subsection{A physical motivation}
Since DNA molecules and other biological fibers are constantly subjected to forces and torques during multiple important biological processes (e.g., transcription, replication and DNA packaging), it is of paramount importance to develop a simple and robust mathematical framework capable of capturing their essential mechanical response. For double-stranded DNA (dsDNA) molecules, certain piecewise defined elastic rod models (so called \emph{worm-like-chain models}) fit the experimental data well for isolated tensile forces below approximately 50 pN; see \cite{MarkoSiggia95, Smith96, Wang97, Bust2000}. Here, the continuum ``rod" is the virtual circular cylinder that the the double helix wraps around. As the tensile force enters the 50-60 pN range, the length of the molecule saturates near its contour length $L_0$, i.e., the longest theoretical length of the molecule. A class of elastic rod models capable of appropriately capturing the stretch-limiting behavior exhibited by dsDNA and other biological fibers was recently introduced by the authors in \cite{RAJROD22A}.    

The assumption that DNA can be described by an elastic rod is a rough approximation, and in most applications it is more appropriate to describe it using viscoelastic constitutive relations. For example, when applied tensile forces surpass approximately 65 pN, optical tweezers experiments on individual dsDNA molecules show that the dsDNA molecule undergoes a hysteretic and rate-dependent \emph{overstretching transition}; see \cite{Smith96, Wang97, RouzinaBloomfield2001}. The molecule stretches to approximately $1.8 L_0$, and its response approaches that of a single-stranded DNA molecule (ssDNA) as the force increases. After many years of debate and experiments, the main consensus appears to be that the overstretching transition is the product of three distinct molecular mechanisms: strand splitting, base-pair bonds melting, and conversion into ladder like structures (S-form DNA); see the reviews \cite{ZaltronetalReview, BustamantePrimer}. How exactly the strand splitting, bond melting and S-form conversion mechanisms are distributed throughout the numerous base pairs of the molecule can be quite challenging to track during the overstretching transition. The work \cite{Grossetal2011} presents another piecewise defined elastic rod model that better takes into account the coupling between twist and stretch and fits the experimental data well up to and including certain overstretching transitions. 

Although useful for fitting certain experimental data, the previously cited continuum rod models disregard the clear rate-dependent viscoelastic behavior exhibited by DNA (during, e.g., the overstretching transition) in favor of elastic rod models. In contrast to elastic rods that are characterized by a single fixed \emph{natural configuration}, the natural configuration of viscoelastic and inelastic rods can evolve when the body is undergoing a thermodynamic process.\footnote{One may
	view the \emph{natural configuration} of a body to be the configuration it
	would take upon the removal of all external stimuli. The notion of an evolving natural configuration was first discussed by Eckart \cite{EckartThermo48}; see Section 3 for more details on this notion for rods and \cite{RajagopalReport695} for more on this notion in continuum mechanics.} The evolution of the natural configuration is determined by the way in which energy is stored and entropy is produced by the rod, and the study of the overstretching transition of DNA and other viscoelastic phenomena would have to take this evolution into consideration.

With this motivation in mind, this work introduces a class of intrinsic, nonlinear, rate-dependent viscoelastic Cosserat rod models incorporating \emph{evolving natural configurations}. Moreover, these mathematical models are \emph{thermodynamically consistent}; the rate of total entropy production during an admissible motion of the rod is non-negative. Since we feel that our interpretation of thermodynamic consistency is much closer to the Second Law of Thermodynamics than the stronger point-wise Clausius-Duhem inequality commonly adopted, a brief discussion of the Second Law is warranted.

\subsection{Second Law of Thermodynamics}

For a process to be thermodynamically allowable, the process has to meet the restrictions imposed by the Second Law of Thermodynamics in its global form. As Eddington \cite{Eddington14} aptly observes, 
\begin{quotation}
\emph{The law that entropy always increases - the Second Law of Thermodynamics - holds, I think, the supreme position among the laws of Nature. If someone points out to you that your pet theory of the universe is in disagreement with Maxwell's equation - then so much the worse for Maxwell's equations. If it is found to be contradicted by observation - well these experimentalists do bungle things sometimes. But if your theory is found to be against the Second Law of Thermodynamics, I can give you no hope; there is nothing for it but to collapse in deepest humiliation...At present we can see no way in which an attack on the Second Law of Thermodynamics could possibly succeed, and I confess personally, I have no great desire that it should succeed in averting the final running down of the universe.}
\end{quotation}
There are various interpretations of the Second Law of Thermodynamics and we cannot discuss their equivalence or otherwise here. We briefly provide the main statements of the Second Law of Thermodynamics that are in place. 

The origins of the Second Law can be traced to the seminal works of Sadi Carnot \cite{Carnot} and it was built by Clausius, Kelvin, Caratheodory, Planck and others.\footnote{The original papers and extracts from the papers by the pioneers of the subject, translated into English, can be found in Magie \cite{Magie} and Kestin \cite{Kestin76}.} Clausius \cite{Kestin76} remarks: 
\begin{quotation}
\emph{Di Energie der Welt ist constant. Die Entropy der Welt strebt einen Maximum zu}.
\end{quotation}
This is usually translated as ``The Energy of the World is constant." The Entropy of the World tends towards a maximum" but often one finds the translation ``The Energy of the Universe is constant. This statement is a global statement concerning the total entropy of the World or Universe, and it does not prohibit the possibility that in some parts of the World or Universe the entropy decreases while in other parts the entropy increases in such a manner that the net entropy can increase. Also, the above statement that the entropy of the World tends towards a maximum does not say that as time tends to infinity, the total entropy of the World tends monotonically towards a maximum, but that is a tacit assumption, and this leads to one assuming that the time derivative of the total entropy of the World is non-negative. The more standard interpretation of the Second Law is that not all of the heat can be converted into work, and that some of the heat is consumed in changing the ``transformation content", namely the entropy of the body.\footnote{Clausius \cite{Clausius62} coined the term ``entropy" to describe the ``transformation content" in a body.} This was expressed by Clausius \cite{Clausius65} as: 
\begin{quotation}
	\emph{-in the production of work it may very well be the case that at the same time a certain quantity of heat is consumed and another quantity transferred from a hotter to a colder body, and both quantities of heat stand in a definite relation to the work that is done.}
\end{quotation}
Kelvin \cite{Thomson} interprets Clausius as stating:
\begin{quotation}
\emph{It is impossible for a self-acting machine, unaided by an external agency, to convey heat from one body to another at a higher temperature...It is impossible, by means of inanimate material agency, to derive mechanical effect from any portion of matter by cooling it below the temperature of the coldest of the surrounding objects.}\footnote{Kelvin \cite{Thomson} adds the following footnote: \emph{If this axiom be denied for all temperatures, it would have to be admitted that a self-acting machine might be set to work and produce mechanical effect by the sea or earth, with no limit but the total loss of heat from the earth and sea, or, in reality, from the whole material world.}}
\end{quotation}
Caratheodory \cite{Caratheodory09} did not provide a statement of the Second Law in terms of heat being converted to work. Instead, he remarks:
\begin{quotation}
\emph{In every arbitrarily close neighborhood of a given initial state there exists states which cannot be approached arbitrarily closely by adiabatic processes.}
\end{quotation}
Planck \cite{Planck} interprets the Second Law as:
\begin{quotation}
\emph{It is impossible to construct an engine which will work in a complete cycle, and produce no effect except of raising of a weight and cooling of a heat reservoir.}
\end{quotation}
It is important to recognize that all of these early works on the Second Law were all concerned with the conversion of heat into work. It was only later that the notion of ``entropy" was associated with the statistical explanation of the extent of disorder in a system. In fact, the early notions of the Second Law were within the context of classical thermodynamics wherein there is no concept of a field; notions such as temperature, energy, and entropy were each a single number associated with a homogeneous body as a whole. 

Today, the Second Law is interpreted as not only being applicable to the universe as a whole but also to an isolated system whose total entropy (a primitive concept) does not decrease when it undergoes a thermodynamic process. The entropy remains the same in a reversible process and is positive when it undergoes an irreversible process. The isolated system tends to equilibrium, and the entropy of the isolated system tends to a maximum, asymptotically in time. In the case of continuum thermodynamics, once again the Second Law is assumed to apply to an isolated system.  
Truesdell and Muncaster \cite{TruesdellMuncaster} state that:
\begin{quotation}
	\emph{The objectives of continuum thermomechanics stop far short of explaining the universe, but within that theory we may easily derive an explicit statement in some ways reminiscent of Clausius, but referring only to a modest object: an isolated body of finite size.}
\end{quotation}
However, in continuum thermodynamics, the Clausius-Duhem inequality takes the role of the Second Law (see Truesdell \cite{Truesdell52}), and it is invariably enforced in the local point-wise form. The global form of the Clausius-Duhem inequality states that the time derivative of the total entropy of the body, as a whole, at any time is greater than or equal to the sum of the entropy flux into the body due to heat flux and entropy flux into the body due to radiation: for all $t$, 
\begin{align}
	\frac{d}{dt} \int_{\kappa_t(\cl B)} \rho \eta  \,dv \geq -\int_{\p \kappa_t(\cl B)} \frac{\bs q \cdot \bs n}{\theta} da + \int_{\p \kappa_t(\cl B)} \frac{r}{\theta} dv. \label{eq:globalen}
\end{align} 
In \eqref{eq:globalen}, $\kappa_t(\cl B)$ is the current configuration of the three-dimensional body $\cl B$ at time $t$, $\bs n$ is the outward pointing unit normal to the boundary $\p \kappa_t(\cl B)$, $\rho$ is the density, $\eta$ is the specific entropy, $\bs q$ is the heat flux, and $r$ is the radiant heating.

The requirement that \eqref{eq:globalen} hold with $\kappa_t(\cl B)$ replaced by an arbitrary sub-part $\cl P_t \subseteq \kappa_t(\cl B)$, i.e.,
\begin{align}
	\frac{d}{dt} \int_{\cl P_t} \rho \eta  \,dv \geq -\int_{\p \cl P_t} \frac{\bs q \cdot \bs n}{\theta} da + \int_{\p \cl P_t} \frac{r}{\theta} dv \label{eq:localen}, \quad \forall \cl P_t \subseteq \kappa_t(\cl B), 
\end{align} 
is much more restrictive than \eqref{eq:globalen}. Assuming appropriate smoothness of the variables, the more restrictive requirement \eqref{eq:localen} is equivalent to the point-wise inequality: for all $\bs x \in \kappa_t(\cl B)$ and $t$,
\begin{align}
	\rho \dot \eta \geq -\div \Bigl (
	\frac{\bs q}{\theta}
	\Bigr ) + \frac{\tau}{\theta}, \label{eq:CDlocal}
\end{align} 
where $\dot{\mbox{}} = \frac{D}{Dt}$ is the material time derivative. Green and Naghdi \cite{GreenNaghdi77a} rendered the local form of the Clausius-Duhem inequality \eqref{eq:CDlocal} as an equality by introducing a term $\xi$ that signifies the specific rate of entropy production: 
\begin{align}
	\rho \xi := \rho  \dot \eta + \div \Bigl (\frac{\bs q}{\theta} \Bigr ) - \frac{r}{\theta},
\end{align} 
which they required to be non-negative, $\xi \geq 0$. Later Ziegler and collaborators (see \cite{Ziegler63, Ziegler72, ZieglerWeh87}) introduced two different requirements that essentially impose that the rate of dissipation must be orthogonal to the level surfaces of a dissipation function and that the rate of dissipation should be maximized. Since Ziegler \cite{Ziegler63} presumes that the entropy production is determined only by the velocities, and that forces have to be fixed during the maximization, his ideas are incapable of describing many phenomena that are observed. A detailed critical discussion can be found in the works by the first listed author and Srinivasa \cite{RajSriPRSL05, RajSriZAMP08, RajSriIJES19}. Nonetheless, the ideas introduced by Ziegler and co-workers was an important step on the development of constitutive relations in continuum thermodynamics.  

In \cite{RAJAGOPALSRI2000}, the first author and Srinivasa assumed forms for the Helmholtz potential and the point-wise rate of entropy production and enforced the maximization of the rate of entropy production to obtain constitutive relations for non-Newtonian fluids. Later, the first author and Srinivasa \cite{RajSriPRSL11} assumed forms for the Gibbs potential and the point-wise rate of entropy production, and once again required the maximization of entropy production to obtain constitutive relations. These works subsequently led to numerous studies for the development of constitutive relations in a wide swathe of areas, but all within the context that maximization be demanded at every point in the body; see e.g. \cite{RajSriZAMP04a, RajSriZAMP04b, RajSriBending05, RajSriJNNFM01, RaoRaj2002, KannanRaoRaj08, KannanRaj11, MalRajTum, BarotRao, MoonCuiRao, MalekPrusaSkirvanSuli18, BathBulMal21, BulMalRod22, PrusaRajTuma20, PrasadMoonRao21, SodCruRao15, SreeKannRaj21, SreeKannRaj23a, RajRod22,  SreeSrikKannRaj23}. 

\subsection{Main results and outline}
Using ideas from \cite{RAJAGOPALSRI2000}, this work introduces a class of intrinsic, nonlinear, rate-dependent viscoelastic Cosserat rod models incorporating the notion that natural configurations evolve. In the (special) Cosserat theory, the \emph{current configuration} of the rod at time $t$, is modeled by a one-dimensional curve in Euclidean space, 
\begin{align}
	[0,L] \ni s \mapsto \bs r(s,t), 
\end{align}
(the \emph{center line}) to which a right-handed collection of orthonormal vector fields $ \{ \bs d_k(\cdot) \}_{k = 1}^3$ (the \emph{directors}) is attached; see Figure \ref{fig:1}. 

\begin{figure}[t]
	\begin{center}
	\includegraphics[scale=.60]{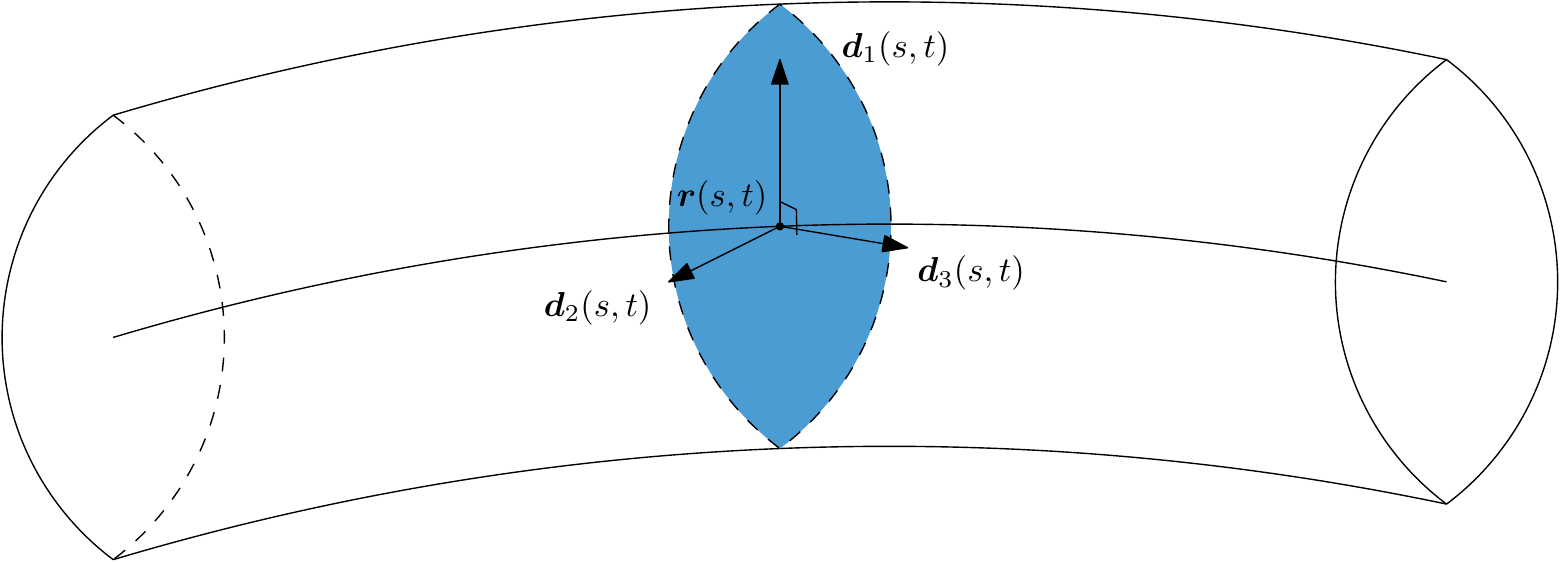}
	\caption{The current configuration of a special Cosserat rod.}
		\label{fig:1}
	\end{center}
\end{figure}

The directors $\{\bs d_1(s, t), \bs d_2(s, t)\}$ portray the orientation of the cross section transverse to the center line at $\bs r(s, t)$,   
can deform independently of the center line, and are specified by the Darboux vector field 
\begin{align}
	\bs u(s,t) = \frac{1}{2} \bs d_k(s,t) \times \p_s \bs d_k(s,t), \quad s \in [0,L],
\end{align}  
where here and throughout this work, we use the Einstein summation convention.  
The six components of $\bs u(s,t)$ and $\bs v(s,t) = \p_s \bs r(s,t)$ in the frame $\{\bs d_k(s,t)\}_{k = 1}^3$ are the local measures of strain in the standard theory. We then introduce primitive variables that express the resultant contact forces, contact couples, heat flux, internal energy and entropy of the rod and postulate balance laws of linear momentum, angular momentum, energy and entropy (see Section 2.2). In particular, the form of the Second Law required during the motion of the rod is that the rate of total entropy production is non-negative (see \eqref{eq:entropyeq}). In the isothermal setting that we consider, this requirement is equivalent to the inequality
\begin{align}
\int_0^L \bigl [\bs m \cdot \p_s \bs w + \bs n \cdot \p_{st}^2 \bs r - (\p_s \bs r \times \bs n) \cdot \bs w - \p_t \psi \bigr] ds \geq 0, \label{eq:diss2}  
\end{align}
where $\bs w = \frac{1}{2} \bs d_k \times \p_t \bs d_k$, $\bs m$ and $\bs n$ are the contact couple and contact force vector fields, and $\psi$ is the \emph{Helmholtz free energy} (to be specified a priori). 

In our model, we also include a second Darboux vector field $\bs u_d(\cdot,t)$ and tangent vector field $\bs v_d(\cdot, t)$ that specify the \emph{natural configuration} of the rod at time $t$ (see Section 3.1). The field equations are closed upon specifying \emph{constitutive relations} between $\bs m$, $\bs n$, $\bs u$, $\bs v$, $\bs u_d$ and $\bs v_d$ and postulating two evolution equations for $\bs u_d$ and $\bs v_d$. Following \cite{RAJAGOPALSRI2000}, we first propose frame indifferent forms of the Helmholtz free energy and rate of total entropy production (\emph{two scalar functions}) that is consistent with \eqref{eq:diss2} (see Section 3.2 and Section 4.1). 

In contrast to \cite{RAJAGOPALSRI2000} and the references listed in Section 1.2, we then obtain constitutive relations and evolution equations for the natural configuration by requiring the motion to maximize the \emph{rate of total entropy production} rather than point-wise rate of entropy production. Depending on the space of functions where maximization is required, in Section 4.2 we obtain one model wherein the integrand from \eqref{eq:diss2} is point-wise non-negative (see \eqref{eq:mconsteq}, \eqref{eq:nconsteq}, \eqref{eq:unatevol}, \eqref{eq:vnatevol}) and a second distinct model wherein the natural strain variables are uniform-in-$s$ (see \eqref{eq:mconsteq}, \eqref{eq:nconsteq}, \eqref{eq:unatevolunif}, \eqref{eq:vnatevolunif}). In the second development of the constitutive relations for the rod, \eqref{eq:diss2} still holds but the point-wise rate of entropy production is not necessarily non-negative (see Section 4.3).  

Finally, in Section 5, we consider the closed set of field equations for the simple problem of isolated torsion and a quadratic strain energy. In particular, we show that in contrast to other nonlinear, viscoelastic Cosserat rod models (see \cite{AntmanSeidman05, AntmanBook, LinnLangTuganov}), our model predicts \emph{both} of the common viscoelastic phenomena of solid-like stress relaxation (Section 5.2) and solid-like creep (Section 5.3).\footnote{The models introduced in \cite{AntmanSeidman05, AntmanBook, LinnLangTuganov} are able to incorporate solid-like creep behavior.}  

\section{Preliminaries}

\subsection{Kinematics}
Let $\bbE^3$ be three-dimensional Euclidean space with associated translation space identified with $\bbR^3$.  
Let $\{\bs e_k \}$ be a fixed right-handed orthonormal basis for $\bbR^3$.  The \emph{current configuration} of a special Cosserat rod at time $t$ is defined by a triple: 
\begin{align}
[0,L] \ni s \mapsto (\bs r(s,t), \bs d_1(s,t), \bs d_2(s,t)) \in \bbR^3 \times \bbR^3 \times \bbR^3,
\end{align}
with $\bs d_1(s,t)$ and $\bs d_2(s,t)$ orthonormal for each $s$ and $t$. The element $\bs r(s,t)$ is the position vector, relative to a fixed origin $\bs o \in \bbE^3$, for the \emph{center line} of the configuration at time $t$. The vectors $\{\bs d_1(s,t), \bs d_2(s,t)\}$ are the \emph{directors} and are regarded as tangent to the material cross section transversal to $\bs r(\cdot,t)$ at $\bs r(s,t)$ (see Figure \ref{fig:1}). We view $[0,L]$ as parameterizing the material points of the center line at each time $t$.

Let
\begin{align*}
	\bs d_3(s,t) = \bs d_1(s,t) \times \bs d_2(s,t).
\end{align*}
Then $\{ \bs d_k(s,t) \}$ is a right-handed orthonormal basis for $\bbR^3$ describing a configuration of the material cross section at $\bs r(s,t)$, and  
\begin{align}
	\bs R(s,t) &= \bs d_k(s,t) \otimes \bs e_k
\end{align}
is the unique proper rotation satisfying
\begin{align}
	\bs R(s,t) \bs e_k = \bs d_k(s,t), \quad k = 1,2,3. \label{eq:rotdir}
\end{align} 
Conversely, a proper rotation $\bs R(s,t)$ for each $(s,t)$ specifies directors $\{ \bs d_k(s,t)\}$ via \eqref{eq:rotdir}. Since the center line of a configuration can always be obtained uniquely (up to a spatial translation) from its tangent vector field, a configuration can be equivalently defined by a pair: 
\begin{align}
	[0,L] \ni s \mapsto (\p_s \bs r(s,t), \bs R(s,t)) \in \bbR^3 \backslash \{\bs 0\} \times SO(3). 
\end{align}

Since $\bs R(s,t)$ is a rotation for each $(s,t)$,
there exist unique axial vectors
\begin{align*}
	\bs u(s,t) &= u_k(s,t) \bs d_k(s,t) \in \bbR^3, \\
	\bs w(s,t) &= w_k(s,t) \bs d_k(s,t) \in \bbR^3,
\end{align*} 
such that for all $\bs a \in \bbR^3$, 
\begin{align}
	[\p_s \bs R(s,t)] \bs R^*(s,t) \bs a &= \bs u(s,t) \times \bs a, \\	
	[\p_t \bs R(s,t)] \bs R^*(s,t) \bs a &= \bs w(s,t) \times \bs a.
\end{align}
In particular, it follows that for $k = 1,2,3$,
\begin{align}
  \p_s \bs d_k(s,t) &= \bs u(s,t) \times \bs d_k(s,t), \\
  \p_t \bs d_k(s,t) &= \bs w(s,t) \times \bs d_k(s,t), 
\end{align}
or equivalently, 
\begin{align}
  \bs u(s,t) &= \frac{1}{2} \bs d_k(s,t) \times \p_s \bs d_k(s,t), \\
 \bs w(s,t) &= \frac{1}{2} \bs d_k(s,t) \times \p_t \bs d_k(s,t).
\end{align}
The components $u_1$ and $u_2$ are referred to as the \emph{flexural strains}, and the component $u_3$ is referred to as the \emph{torsional strain} (or \emph{twist}). The vector $\bs w(s,t)$ is referred to as the resultant \emph{angular velocity} of the material section at $\bs r(s,t)$. 

 In what follows we will denote the tangent vector field $\p_s \bs r(\cdot,t)$ by
 $\bs v(\cdot,t)$:  
\begin{align}
	\bs v(s,t) = \p_s \bs r(s,t) = v_k(s,t) \bs d_k(s,t). 
\end{align}  
The components $v_1$ and $v_2$ are referred to as the \emph{shear strains}. The component $v_3$ is referred to as the \emph{dilation strain}, and an orientation of the director $\bs d_3(\cdot,t)$ relative to the center line $\bs r(\cdot,t)$ is fixed by requiring that every  configuration satisfies, for all $(s,t)$, 
\begin{align}
 v_3(s,t) > 0. \label{eq:v3pos}
\end{align} 
The restriction \eqref{eq:v3pos} also implies that the \emph{stretch} of the rod in every  configuration is never zero, $|\p_s \bs r(s,t)| > 0$, and that the rod cannot be sheared so severely that a section becomes tangent to the center line.

\subsection{Balance laws}

Let $[a,b] \subseteq [0,L]$ and $t \in [0,\infty)$. We denote the contact force by $\bs n(s,t)$ so that the resultant force on the material segment $[a,b]$ by $[0,a) \cup (b,L]$ at time $t$ is given by 
\begin{align}
	\bs n(b,t) - \bs n(a,t).
\end{align}
The contact couple is denoted by $\bs m(s,t)$ so that the resultant contact couple about $\bs o \in \bbE^3$ on the material segment $[a,b]$ by $[0,a) \cup (b,L]$ is given by 
\begin{align}
	\bs m(b,t) + \bs r(b,t) \times \bs n(b,t) - \bs m(a,t) - \bs r(a,t) \times \bs n(a,t). 
\end{align}
For each $(s,t)$, the contact force and contact couple may be expressed in the director frame $\{\bs d_k(s,t)\}$ via 
\begin{align}
	\bs n(s,t) = n_k(s,t) \bs d_k(s,t), \quad \bs m(s,t) = m_k(s,t) \bs d_k(s,t). 
\end{align}
The components $n_1$ and $n_2$ are referred to as the \emph{shear forces}, and the component $n_3$ is referred to as the \emph{tension} (or \emph{axial force}). The components $m_1$ and $m_2$ are referred to as the \emph{bending couples} (or \emph{bending moments}), and the component $m_3$ is referred to as the \emph{twisting couple} (or \emph{twisting moment}).

We denote the internal energy per unit reference length and entropy per unit reference length of the rod by $e(s,t)$ and $\eta(s,t)$ respectively. The heat flux is denoted by $q(s,t)$ so that the resultant heat supplied to the material segment $[a,b]$ by $[0,a) \cup (b,L]$ is given by 
\begin{align}
	-(q(b,t) - q(a,t)).
\end{align}
We denote the absolute temperature by $\theta(s,t) > 0$. Similar to Green and Nagdhi \cite{GreenNaghdi77a}, we introduce a variable $\xi(s,t)$ so that the rate of entropy production for the material segment $[a,b]$ is given by 
\begin{align}
	\int_a^b \frac{\xi(s,t)}{\theta(s,t)} ds.
\end{align}
A form of the Second Law of Thermodynamics that we adopt is requiring that the rate of total entropy production is always non-negative: for all $t$ 
\begin{align}
	\int_0^L \frac{\xi(s,t)}{\theta(s,t)} ds \geq 0. \label{eq:dissineq}
\end{align}

Let $(\rho A)(s)$ be the mass per unit reference length of the rod (given a priori). We assume that $\{ \bs e_1, \bs e_2 \}$ are aligned along the undeformed rod's material sections' principal axes of inertia. Then the moment of inertia tensor $(\rho \bs J)(s,t)$ at $(s,t)$ is given by
\begin{align}
	(\rho \bs J)(s,t) &= I_{22}(s) \bs d_1(s,t) \otimes \bs d_1(s,t) + I_{11}(s) \bs d_2(s,t) \otimes \bs d_2(s,t) \\&\quad + (I_{11}(s) + I_{22}(s)) \bs d_3(s,t) \otimes \bs d_3(s,t), 
\end{align}
where $\{I_{\mu \mu}(s)\}$ are the second mass moments of inertia of the material section at $s$ (given a priori). 

We omit the dependence on $(s,t)$ in what follows. If $\bs f$ is an external body force per unit reference length, $\bs l$  is an external body couple per unit reference length, and $r$ is an external heat source per unit reference length, then the classical equations expressing balance of linear momentum, angular momentum, energy and total entropy are given by:
\begin{align}
	(\rho A) \p_t^2 \bs r &= \p_s \bs n + \bs f, \label{eq:linmom} \\
	\p_t \bigl [(\rho \bs J) \bs w \bigr ] &= \p_s \bs m + \p_s \bs r \times \bs n + 
	\bs l, \label{eq:angmom} \\
	\p_t e &= \bs n \cdot \p^2_{st} \bs r + \bs m \cdot \bs \p_s \bs w - 
	(\p_s \bs r \times \bs n) \cdot \bs w \\&\quad + r - \p_s q, \label{eq:eneq}, \\
	 \p_t \eta &= \frac{r}{\theta} - \p_s \frac{q}{\theta} + \frac{\xi}{\theta}. \label{eq:entropyeq} 
\end{align}
For more on their derivations, we refer the reader to Chapter 8 of \cite{AntmanBook} for \eqref{eq:linmom}-\eqref{eq:angmom} and to \cite{GreenLaws66, DeSilvaWhitman71, SmritiKumGroStei19} for \eqref{eq:eneq} and \eqref{eq:entropyeq}. 

In the remainder of this work we will only consider the isothermal setting: for all $(s,t)$
\begin{align}
	q(s,t) = 0, \quad \theta(s,t) = \theta_0 > 0.
\end{align}
Let 
\begin{align}
	\psi = e - \theta \eta,
\end{align}
the \emph{Helmholtz free energy} of the rod. The relations \eqref{eq:eneq} and \eqref{eq:entropyeq} imply that the point-wise, re-scaled rate of entropy production $\xi$ satisfies 
\begin{align}
	\xi = \bs m \cdot \p_s \bs w + \bs n \cdot \p_{st}^2 \bs r - (\p_s \bs r \times \bs n) \cdot \bs w - \p_t \psi. \label{eq:disspw}
\end{align}
Then the re-scaled rate of total entropy production satisfies
\begin{align}
\int_0^L \xi \, ds = \int_0^L \bigl [\bs m \cdot \p_s \bs w + \bs n \cdot \p_{st}^2 \bs r - (\p_s \bs r \times \bs n) \cdot \bs w - \p_t \psi \bigr] ds. \label{eq:diss1}
\end{align}
In what follows, we will refer to $\int_0^L \xi \, ds$ as the \emph{total dissipation rate}.  

\section{Evolving natural configurations}

\subsection{Kinematics of the natural configuration}

At each time $t$, we denote a second configuration by 
\begin{align}
	[0,L] \ni s \mapsto (\bs v_d(s,t), \bs R_d(s,t)) \in \bbR^3 \backslash \{\bs 0 \} \times SO(3),
\end{align}
referred to as a \emph{natural configuration}. The 
axial vector for $[\p_s\bs R_d(s,t)] \bs R_d^*(s,t)$ is denoted by $\bs u_d(s,t)$, and the associated natural directors are denoted by 
\begin{align}
\bs d_{d,k}(s,t) = \bs R_d(s,t) \bs e_k. \label{eq:dp}
\end{align} 
Then we have expansions
\begin{align}
	\bs u_d(s,t) &= u_{d,k}(s,t) \bs d_{d,k}(s,t), \label{eq:up} \\
	\bs v_d(s,t) &= v_{d,k}(s,t) \bs d_{d,k}(s,t). \label{eq:vp}
\end{align}

A natural configuration at time $t$ is viewed as a configuration the rod could take if all external stimuli are removed by time $t$. How the load is removed, e.g., instantaneously, very slowly, or otherwise, will lead to different natural configurations. In this work we are interested in a viscoelastic rod with instantaneous elastic response relative to the natural configuration. We will discuss the precise mathematical interpretation of these physical assumptions in Section \ref{s:natural}, but, for now, we will focus only on kinematic properties of a natural configuration. The role of natural configurations in general is discussed in detail in \cite{RajagopalReport695}.

Let $(\bs v(\cdot,t) = \p_s \bs r(\cdot,t), \bs R(\cdot,t))$ specify the current configuration of the rod at time $t$. For each $(s,t) \in [0,L] \times [0,\infty)$, we define a pair
\begin{align}
(\bs v_e(s,t), \bs R_e(s,t)) \in \bbR \times \bbR^3 \times SO(3), 
\end{align}
by 
\begin{align}
\bs R_e(s,t) &= \bs R(s,t) \bs R_d^*(s,t), \label{eq:Redef} \\
\bs v_e(s,t) &= \bs v(s,t) - \bs R_e(s,t) \bs v_d(s,t). \label{eq:vedef} 
\end{align} 
Then the variables specifying the current configuration admit the decomposition 
\begin{align}
	\bs v(s,t) &= \bs R_e(s,t) \bs v_d(s,t) + \bs v_e(s,t), \label{eq:vdecomp} \\
	\bs R(s,t) &= \bs R_e(s,t) \bs R_d(s,t). \label{eq:Rdecomp}
\end{align}
Let $\bs u_e(s,t)$
be the axial vector of $[\p_s \bs R_e(s,t)] \bs R_e(s,t)^*$. Since  
\begin{align}
	[\p_s \bs R(s,t)] \bs R^*(s,t) = \bs R_e(s,t)[\p_s \bs R_d(s,t) \bs R^*(s,t)] \bs R_e(s,t)^* + \bs R_e(s,t) \bs R_e(s,t)^*,
\end{align}
we conclude that 
\begin{align}
	\bs u(s,t) = \bs R_e(s,t) \bs u_d(s,t) + \bs u_e(s,t). \label{eq:udecomp}
\end{align}

\begin{figure}[t]
	\centering
	\includegraphics[width=\linewidth]{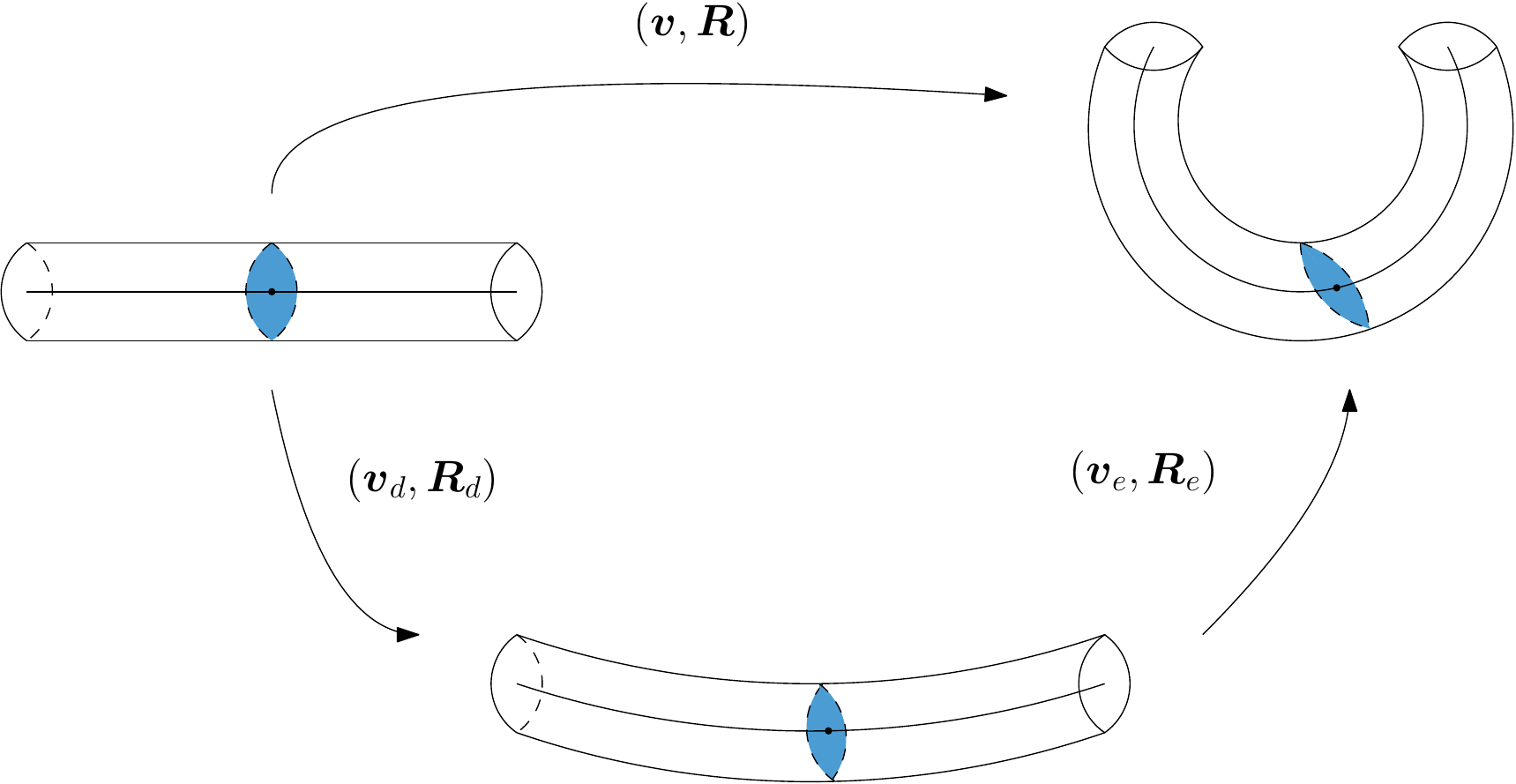}
	\caption{A schematic of the two configurations of the rod. The current configuration of the rod may be thought of as being composed of a ``dissipative" process $(\bs v_d, \bs R_d)$ and an ``elastic" process $(\bs v_e, \bs R_e)$ (see \eqref{eq:ustraindecomp}, \eqref{eq:vstraindecomp}).}
	\label{fig:2}
\end{figure}

In summary, the vector field 
\begin{align}
\bs u_e(\cdot,t) = u_{e,k}(\cdot,t) \bs d_k(\cdot,t) \label{eq:ue}
\end{align} is the axial vector field generating the rotation field $\bs R_e(\cdot,t)$ which, for each $s \in [0,L]$, rotates the natural director frame $\{ \bs d_{d,k}(s,t) \}$ to the current director frame $\{ \bs d_{k}(s,t) \}.$ The vector field 
\begin{align}
\bs v_e(\cdot,t) = v_{e,k}(\cdot,t) \bs d_k(\cdot,t) \label{eq:ve}
\end{align}
 is the difference between the rotated natural tangent vector field $\bs R_e(\cdot,t) \bs v_d(\cdot,t)$ and the current tangent vector field $\bs v(\cdot,t)$. Therefore, we may interpret the vector fields $\bs u_e(\cdot,t)$ and $\bs v_e(\cdot,t)$ as locally characterizing the deformation from the natural configuration to the current configuration. We conclude this subsection by noting that 
 \eqref{eq:dp}, \eqref{eq:up}, \eqref{eq:vp} and \eqref{eq:Rdecomp}  imply that
 \begin{align}
 	\bs R_e \bs u_{d} = u_{d,k} \bs d_k, \quad \bs R_e \bs v_d = v_{d,k} \bs d_k,
 \end{align}
and thus, by \eqref{eq:udecomp} and \eqref{eq:vdecomp}, for each $k = 1,2,3$,
\begin{align}
	u_k = u_{d,k} + u_{e,k}, \label{eq:ustraindecomp}\\
	v_k = v_{d,k} + v_{e,k}. \label{eq:vstraindecomp}
\end{align}

\subsection{Frame indifference}\label{s:changeinobserver}

In what follows we omit the dependence on $(s,t)$. 
\begin{defn}
 We say that the pair of  configurations $((\bs v_d, \bs R_d), (\bs v, \bs R))$, with associated director fields $\{ \bs d_{d,k} \}$ and $\{\bs d_{k}\}$, are \emph{related by a change of frame} to the 
pair of  configurations $((\tbsv_d, \tbsR_d),(\tbsv, \tbsR))$, with associated director fields $\{ \tbsd_{d,k} \}$ and $\{\tbsd_{k}\}$,
if there exists $\bs Q(t): (-\infty,\infty) \rar SO(3)$ such that 
\begin{align}
		\tbsv_d &= \bs Q \bs v_d, \label{eq:vptilde} \\
	\tbsd_{d,k} &= \bs Q \bs d_{d,k}, \quad k = 1,2,3, \label{eq:dptilde} \\
	\tbsv &= \bs Q \bs v, \label{eq:vtilde} \\
	\tbsd_k &= \bs Q \bs d_k, \quad k = 1,2,3. \label{eq:dtilde}
\end{align}
\end{defn}

Suppose that $((\bs v_d, \bs R_d), (\bs v, \bs R))$ and $((\tbsv_d, \tbsR_d), (\tbsv, \tbsR))$ are related by a change of frame. By \eqref{eq:rotdir} we see that \eqref{eq:dptilde} and \eqref{eq:dtilde} imply that 
\begin{align}
	\tbsR_d = \bs Q \bs R_d, \quad \tbsR = \bs Q \bs R, \\
		\tbsR_e = \tbsR \tbsR_d^* = \bs Q \bs R_e \bs Q^*. \label{eq:tRe}
\end{align}
We conclude that 
\begin{gather}
\begin{split}
		\tbsR_e \tbsu_d = \bs Q \bs u_d, \quad \tbsR_e \tbsv_d = \bs Q \bs v_d, \\
	\tbsu = \bs Q \bs u, \quad \tbsv = \bs Q \bs v,
	\end{split}\label{eq:tildetrans}
\end{gather}
and, in particular, 
\begin{gather}
	\begin{split}
	\tbsR^* (\tbsR_e \tbsu_d) = \bs R^*(\bs R_e \bs u_d) = u_{d,k} \bs e_k, \\ 
	\tbsR^* (\tbsR_e \tbsv_d) = \bs R^*(\bs R_e \bs v_d) = v_{d,k} \bs e_k, \\
	\tbsR^* \tbsu = \bs R^* \bs u = u_k \bs e_k, \\
	\tbsR^* \tbsv = \bs R^* \bs v = v_k \bs e_k.
\end{split}\label{eq:tildetransformationsR} 
\end{gather}

For what follows we denote 
\begin{gather}
	\bs X = SO(3) \times \bbR^3 \times \bbR^3  \backslash \{\bs 0\}  \times \bbR^3 \times \bbR^3  \backslash \{\bs 0\}, \\
	\msX = \bbR^3 \times \bbR^3  \backslash \{\bs 0\}  \times \bbR^3 \times \bbR^3  \backslash \{\bs 0\}, \\
	\bs Y = \bs X \times \bbR^3 \times \bbR^3 \times \bbR^3 \times \bbR^3, \\
	\msY = \msX \times \bbR^3 \times \bbR^3 \times \bbR^3 \times \bbR^3. 
\end{gather}
\begin{defn}
	We say that scalar-valued functions  
	\begin{gather}
		 \psi : \bs X \rar \bbR, \quad 
		\cl F : C([0,L]; \bs Y)\rar \bbR,  
	\end{gather}
	are \emph{frame indifferent} if for all pairs of  configurations $$((\bs v_{d}, \bs R_d), (\bs v, \bs R)) \mbox{ and } ((\tbsv_{d}, \tbsR_d), (\tbsv, \tbsR))$$ related by a change of frame, we have  
	\begin{gather}
		\psi(\tbsR, \tbsR_e \tbsu_d, \tbsR_e \tbsv_d, \tbsu, \tbsv) = \psi(\bs R, \bs R_e \bs u_d, \bs R_e \bs v_d, \bs u, \bs v), \\
		\cl F(\tbsR, \tbsR_e \tbsu_d, \tbsR_e \tbsv_d, \tbsu, \tbsv, \p_t [\tbsR_e \tbsu_d], \p_t [\tbsR_e \tbsv_d], \p_t \tbsu, \p_t \tbsv) \\
		= \cl F(\bs R, \bs R_e \bs u_d, \bs R_e \bs v_d, \bs u, \bs v, \p_t [\bs R_e\bs u_d], \p_t [\bs R_e \bs v_d], \p_t \bs u, \p_t \bs v). 
	\end{gather}
Here the argument of the function $\psi$ is the five-tuple of point-wise values 
\begin{align}
(\bs R(s,t), \bs R_e(s,t) \bs u_d(s,t), \bs R_e(s,t) \bs v_d(s,t), \bs u(s,t), \bs v(s,t)),
\end{align}
and the argument of the functional $\cl F$ is the nine-tuple of {functions}
\begin{gather}
	(\bs R(\cdot,t), \bs R_e(\cdot,t) \bs u_d(\cdot,t), \bs R_e(\cdot,t) \bs v_d(\cdot,t), \bs u(\cdot,t), \bs v(\cdot,t), \\ \p_t [\bs R_e(\cdot,t)\bs u_d(\cdot,t)], \p_t [\bs R_e(\cdot,t) \bs v_d(\cdot,t)], \p_t \bs u(\cdot,t), \p_t \bs v(\cdot,t)).
\end{gather}
\end{defn}

By fairly standard arguments, we obtain the following.

\begin{prop}\label{p:change}
Scalar-valued functions $\psi : \bs X \rar \bbR$ and $\cl F: C([0,L]; \bs Y) \rar \bbR$ are frame indifferent if and only if there exist functions $\hat \psi : \msX \rar \bbR$ and $\hat{\cl F} : C([0,L]; \msY) \rar \bbR$ such that for every  pair of configurations $((\bs v_d, \bs R), (\bs v, \bs R))$ there holds 
\begin{gather}
	\psi(\bs R, \bs R_e\bs u_d, \bs R_e \bs v_d, \bs u, \bs v) = \hat \psi( \mfu_d, \mfv_d, \mfu, \mfv ), \label{eq:psiobs} \\
	\cl F(\bs R, \bs R_e \bs u_d, \bs R_e \bs v_d, \bs u, \bs v, \p_t [\bs R_e\bs u_d], \p_t [\bs R_e \bs v_d], \p_t \bs u, \p_t \bs v) \\ = \hat{\cl F}( \mfu_d, \mfv_d, \mfu, \mfv, \p_t \mfu_d, \p_t \mfv_d, \p_t \mfu, \p_t \mfv), \label{eq:xiobs}
\end{gather} 
where 
\begin{gather}
\mfu_d = \bs R^*(\bs R_e \bs u_d) = \bs R_d^* \bs u_d = u_{d,k} \bs e_k, \quad \mfv_d = \bs R^*(\bs R_e \bs v_d) = \bs R_d^* \bs v_d = v_{d,k} \bs e_k, \\
\mfu = \bs R^* \bs u = u_k \bs e_k, \quad \mfv = \bs R^* \bs v = v_k \bs e_k.
\end{gather}
\end{prop}

\begin{proof}
	Suppose first that $\psi$ is frame indifferent. Let $((\bs v_{d}, \bs R_d), (\bs v, \bs R))$ be an arbitrary pair of configurations, and consider
	\begin{gather}
		\tbsv_{d} = \bs R^* \bs v_{d}, \quad \tbsR_d = \bs R^* \bs R_d, \\
		\tbsv = \bs R^* \bs v, \quad \tbsR = \bs R^* \bs R = \bs I. 
	\end{gather}
	Then by \eqref{eq:tildetrans} we have 
	\begin{align}
		\psi (\bs R, \bs R_e \bs u_d, \bs R_e \bs v_d, \bs u, \bs  v) &= \psi(\bs I, \bs R^*(\bs R_e \bs u_d), \bs R^*(\bs R_e \bs v_d), \bs R^* \bs u, \bs R^* \bs v) \\
		&= \psi(\bs I, \mfu_d, \mfv_d, \mfu, \mfv),
	\end{align}
	 establishing \eqref{eq:psiobs}. Conversely, if $\psi$ satisfies \eqref{eq:psiobs}, then \eqref{eq:tildetransformationsR} implies that $\psi$ is frame indifferent. The proof for $\cl F$ is completely analogous and is omitted. 
\end{proof}

\subsection{In what sense the natural configuration is natural}\label{s:natural}

The precise sense that the variables $\mfu_d$ and $\mfv_d$ model a natural configuration of the rod is expressed via the following assumption on the free energy. 
\begin{assum}\label{ass:natural}
	The Helmholtz free energy $\psi$ satisfies for all $\mfu_d, \mfv_d, \mfu, \mfv$, 
	\begin{align}
		\partial_{\mfu} \psi \big |_{\mfu = \mfu_d, \mfv = \mfv_d} = \bs 0, \quad 
		\partial_{\mfv} \psi \big |_{\mfu = \mfu_d, \mfv = \mfv_d} = \bs 0. \label{eq:natural}
	\end{align}
\end{assum}

For the constitutive equations that we derive in Section 4.2 (see \eqref{eq:mconsteq}, \eqref{eq:nconsteq}), we see that \eqref{eq:natural} is equivalent to the condition that the contact couple and force vanishes when the rod's current configuration is at rest in the natural configuration, i.e., 
\begin{align}
	(\p_t \msu, \p_t \msv, \msu, \msv) = (\bs 0, \bs 0, \msu_d, \msv_d) \implies \msm = \bs 0, \msn = \bs 0.
\end{align}
A simple class of free energies satisfying \eqref{eq:natural} is given by
\begin{align}
	\psi = \frac{1}{2} (\mfu - \mfu_d)\cdot \mfA (\mfu - \mfu_d) + 
	\frac{1}{2} (\mfv - \mfv_d) \cdot \mfB (\mfv - \mfv_d) + \hat \psi_d(\mfu_d, \mfv_d), \label{eq:examplepsi}
\end{align}
where $\mfA$ and $\mfB$ take values in the set of positive definite symmetric tensors. 

\subsection{Inextensible and unshearable rods} \label{s:inextunsh}
\begin{defn}
We say that the rod is \emph{inextensible} if every possible natural configuration and current configuration of the rod satisfy  
\begin{align}
	|\bs v_{d}| = 1, \quad |\bs v| = 1. 
\end{align}
The rod is \emph{unshearable} if every possible natural configuration and current configuration of the rod satisfy
\begin{align}
	v_{d,1} = v_{d,2} = 0, \quad v_{1} = v_2 = 0,
\end{align}
that is, the vector $\bs v_d$ is always parallel to the natural director $\bs d_{d,3} = \bs R_d \bs e_3$ and $\bs v$ is always parallel to the current director $\bs d_{3} = \bs R \bs e_3$.  
\end{defn}

Written differently, the rod is unshearable if every  natural configuration and current configuration of the rod satisfy
\begin{align}
	\bs v_d = |\bs v_d| \bs R_d \bs e_3, \quad
	\bs v = |\bs v| \bs R \bs e_3.
\end{align}

For an unshearable rod, we have
\begin{align}
\bs v &= |\bs v|\bs R \bs e_3
= |\bs v| \bs R_e \bs R_d \bs e_3
= \frac{|\bs v|}{|\bs v_d|}\bs R_e \bs v_d, \label{eq:unshcalc}
\end{align} 
implying that  
\begin{align}
	\bs v_e = \Bigl ( \frac{|\bs v|}{|\bs v_d|} - 1 \Bigr ) \bs R_e \bs v_d. 
\end{align}
In particular, an inextensible and unshearable rod equivalently satisfies 
\begin{align}
	\bs v = \bs R_e \bs v_d \iff \bs v_e = \bs 0.
\end{align}
Thus, an inextensible and unshearable rod with evolving natural configuration is characterized by the constraints on the strains,
\begin{gather}
 v_1 = v_2 = 0, v_3 = 1, \\
 v_{d,1} = v_{d,2} = 0, v_{d,3} = 1, \\
 v_{e,1} = v_{e,2} = v_{e,3} = 0.
\end{gather}

\section{Constitutive theory}

In this section we omit the dependence of various variables on $s$ and $t$. Given an associated natural configuration $(\bs v_d, \bs R_d)$,  the current configuration $(\bs v, \bs R)$, contact couple $\bs m$ and contact force $\bs n$ we define vector fields $\mfm, \mfn, \mfu_d, \mfv_d, \mfu_e, \mfv_e, \mfu, \mfv$ via 
\begin{gather}
\mfm = \bs R^* \bs m = m_k \bs e_k, \quad \mfn = \bs R^* \bs n = n_k \bs e_k, \\
\mfu_d = \bs R^*(\bs R_e \bs u_d) = \bs R_d^* \bs u_d = u_{d,k} \bs e_k, \quad \mfv_d = \bs R^*(\bs R_e \bs v_d) = \bs R_d^* \bs v_d = v_{d,k} \bs e_k, \\
\mfu_e = \bs R^* \bs u_e = u_{e,k} \bs e_k, \quad \mfv_e = \bs R^* \bs v_e = v_{e,k} \bs e_k, \\
\mfu = \bs R^* \bs u = u_k \bs e_k, \quad \mfv = \bs R^* \bs v = v_k \bs e_k, 
\end{gather}
and we note that \eqref{eq:ustraindecomp} and \eqref{eq:vstraindecomp} can be equivalently written as 
\begin{align}
	\mfu = \mfu_d + \mfu_e, \quad \mfv = \mfv_d + \mfv_e. 
\end{align}

In this section we close the field equations for a special Cosserat rod with evolving natural configuration by specifying \emph{constitutive relations} between the variables $\mfm, \mfn, \mfu_d, \mfv_d, \mfu,$ and $\mfv,$ (six 3-dimensional vector fields) and specifying evolution equations for $\mfu_d$ and $\mfv_d$ (two $3$-dimensional vector fields). 

Following \cite{RAJAGOPALSRI2000}, we will accomplish the above tasks in two steps by:
\begin{itemize}
	\item positing forms of the Helmholtz free energy $\psi$ and the total dissipation rate $\int_0^L \xi \, ds$ that are frame indifferent and automatically satisfy \eqref{eq:dissineq} (two \emph{scalar} functions), 
	\item appealing to a variational principle that the strain rates of the rod during its motion maximize the total dissipation rate $\int_0^L \xi \, ds$ (see Proposition \ref{p:maxdiss} and Proposition \ref{p:maxdissuniform} for the precise formulations applied).  
\end{itemize} 

\subsection{Helmholtz free energy and total dissipation rate}

We assume that the Helmholtz free energy $\psi$ is twice continuously differentiable and frame indifferent,
\begin{align}
	\psi = \hat \psi(\mfu_d, \mfv_d, \mfu, \mfv). 
\end{align}  
For simplicity we posit the following explicit, non-negative, rate-dependent, frame-indifferent form of the total dissipation rate
\begin{align}
	\int_0^L \xi \, ds &= \int_0^L \Bigl (\p_t \mfu \cdot \mfM \p_t \mfu 
	+ \p_t \mfv \cdot \mfN \p_t \mfv 
	 \\&\quad + \p_t \mfu_d \cdot \mfM_d \p_t \mfu_d + 
	\p_t \mfv_d \cdot \mfN_d \p_t \mfv_d \Bigr ) ds, \label{eq:xichoice}
\end{align}
where 
\begin{gather}
	\mfM_d = \hat M_{d,k\ell}(\mfu_d, \mfv_d) \bs e_k \tens \bs e_\ell, \quad \mfN_d = \hat N_{d,k\ell}(\mfu_d, \mfv_d) \bs e_k \tens \bs e_\ell, \\
		\mfM = \hat M_{k\ell}(\mfu_d, \mfv_d, \mfu, \mfv) \bs e_k \tens \bs e_\ell, \quad 
			\mfN = \hat N_{k\ell}(\mfu_d, \mfv_d, \mfu, \mfv) \bs e_k \tens \bs e_\ell,
\end{gather}
take values in the set of positive, definite, symmetric tensors. Then \eqref{eq:dissineq} is automatically satisfied. The first two terms appearing on the right hand side of \eqref{eq:xichoice} are standard dissipative terms, while the third and fourth terms appearing on the right hand side of \eqref{eq:xichoice} admit the physical interpretation that \emph{a change in natural configuration always dissipates a nontrivial amount of energy}.
  
\subsection{Maximizing the total dissipation rate}

We now turn to positing relations for the components of the contact couple and contact force in terms of those of the kinematic variables,
\begin{align}
	\begin{split}
	\mfm = \hat \mfm (\mfu_d, \mfv_d, \mfu, \mfv, \p_t \mfu_d, \p_t \mfv_d, \p_t \mfu, \p_t \mfv), \\
	\mfn = \hat \mfn (\mfu_d, \mfv_d, \mfu, \mfv, \p_t \mfu_d, \p_t \mfv_d, \p_t \mfu, \p_t \mfv), 
	\end{split}\label{eq:constequs}
\end{align}  
and evolution equations determining the natural configuration  
\begin{align}
	\p_t \mfu_d = \mfF_d(\mfu_d, \mfv_d, \mfu, \mfv), \quad 
		\p_t \mfv_d = \mfG_d(\mfu_d, \mfv_d, \mfu, \mfv). \label{eq:evoleqs}
\end{align}
We remark that prescribing relations for the components of the contact couple and contact force in terms of the kinematic variables is the ``usual" approach, it goes against causality as the forces and couples are the causes. An alternative approach that is more consistent with causality is prescribing the kinematic variables and their rates in terms of the contact couple and contact force via specifying the Gibb's free energy rather than the Helmholtz free energy (see \cite{RajSriPRSL11}). 

By \eqref{eq:xichoice}, \eqref{eq:diss1} and the relations
\begin{align}
	\p_s \bs w &= (\p_t u_k) \bs d_k,\\
	\p_{st}^2 \bs r &= (\p_t v_k) \bs d_k + \bs w \times \bs v, \\
	(\p_s \bs r \times \bs n) \cdot \bs w &= \bs n \cdot (\bs w \times \bs v),
\end{align}
we conclude that during the motion of the rod, the total dissipation rate satisfies 
\begin{gather}
	\int_0^L \Bigl (\p_t \mfu \cdot \mfM \p_t \mfu + \p_t \mfv \cdot \mfN \p_t \mfv + 
	\p_t \mfu_d \cdot \mfM_d \p_t \mfu_d + \p_t \mfv_d \cdot \mfN_d \p_t \mfv_d \Bigr ) ds = \\
	 \int_0^L \Bigl [(\mfm - \p_{\mfu} \psi) \cdot \p_t \mfu +  
	(\mfn - \p_{\mfv} \psi) \cdot \p_t \mfv
	- \p_{\mfu_d} \psi \cdot \p_t \mfu_d
		- \p_{\mfv_d} \psi \cdot \p_t \mfv_d \Bigr ] ds. 
	 \label{eq:const1} 
\end{gather}   
However, the single scalar constraint \eqref{eq:const1} and knowledge of the Helmholtz free energy $\psi$ do not uniquely determine the sought after relations \eqref{eq:constequs} and evolution equations \eqref{eq:evoleqs},
four point-wise vector equations. Before stating our choice of \eqref{eq:constequs} and \eqref{eq:evoleqs}, we prove the following variational result. 

\begin{prop}\label{p:maxdiss}
	Let $t \in (-\infty, \infty)$ be fixed, and let
	\begin{align}
	\mfm(\cdot,t), \mfn(\cdot,t), \mfu_d(\cdot,t), \mfv_d(\cdot,t), \mfu(\cdot,t), \mfv(\cdot,t) \in L^\infty([0,L]; \bbR^3).
	\end{align}
	 The element 
	 \begin{align}
	 (\dot \mfu_d(\cdot,t), \dot \mfv_d(\cdot,t), \dot \mfu(\cdot,t), \dot \mfv(\cdot,t)) \in L^2([0,L]; (\bbR^3)^4)
	 \end{align}
  determined by the relations
 \begin{align}
 	\mfm &= \p_{\mfu} \hat \psi(\mfu_d, \mfv_d, \mfu, \mfv) + \mfM \dot \mfu, \label{eq:mfmeq} \\
 	\mfn &= \p_{\mfv} \hat \psi(\mfu_d, \mfv_d, \mfu, \mfv) + \mfN \dot \mfv, \label{eq:mfneq} \\
 	\mfM_d \dot \mfu_d &= -\p_{\mfu_d} \hat \psi(\mfu_d, \mfv_d, \mfu, \mfv), \label{eq:dotmfupeq} \\
 	\mfN_d \dot \mfv_d &= -\p_{\mfv_d} \hat \psi(\mfu_d, \mfv_d, \mfu, \mfv), \label{eq:dotmfvpeq}
 \end{align}
	 is the unique maximizer in $L^2([0,L]; (\bbR^3)^4)$ of the functional 
	\begin{gather}
		\cl F(\dot \msu_d, \dot \msv_d, \dot \msu, \dot \msv) = \int_0^L \Bigl ( \dot \mfu \cdot \mfM \dot \mfu 
	+ \dot \mfv \cdot \mfN \dot \mfv 
	+ \dot \mfu_d \cdot \mfM_d \dot \mfu_d + 
	\dot \mfv_d \cdot \mfN_d \dot \mfv_d \Bigr ) ds
	\end{gather}
	 subject to the constraint
	 \begin{gather}
	 		\int_0^L \Bigl (
	 		\dot \mfu \cdot \mfM \dot \mfu + \dot \mfv \cdot \mfN \dot \mfv + 
	 		\dot \mfu_d \cdot \mfM_d \dot \mfu_d + \dot \mfv_d \cdot \mfN_d \dot \mfv_d \Bigr ) ds
	 		= \\ \int_0^L \Bigl [ (\mfm - \p_{\mfu} \psi) \cdot \dot \mfu +  
	 		(\mfn - \p_{\mfv} \psi) \cdot \dot \mfv
	 		 - \p_{\mfu_d} \psi \cdot \dot \mfu_d
	 		- \p_{\mfv_d} \psi \cdot \dot \mfv_d \Bigr ] ds.  
	 	\label{eq:const} 
	 \end{gather}  
For the maximizer, we have
\begin{align}
\cl F(\dot \msu_d, \dot \msv_d, \dot \msu, \dot \msv)
&=\int_0^L \Bigl (
|\mfM^{-1/2} (\mfm - \p_{\mfu} \psi)|^2 + 
|\mfN^{-1/2} (\mfn - \p_{\mfv} \psi)|^2 \\
&\qquad \qquad + |\mfM_d^{-1/2} \p_{\mfu_d} \psi|^2 + |\mfN_d^{-1/2} \p_{\mfv_d} \psi|^2 \Bigr ) ds.  \label{eq:totdissipation} 
\end{align}
\end{prop} 

\begin{proof}
We first note that for any $\mfa, \mfb \in \bbR^3$ and positive definite symmetric matrix $\mfA$, we have 
	\begin{align}
		\mfa \cdot \mfA \mfa &= \mfA^{1/2} \mfa \cdot \mfA^{1/2} \mfa = |\mfA^{1/2} \mfa|^2, \label{eq:A1/2} \\
		\mfa \cdot \mfA \mfb & \leq |\mfA^{-1/2} \mfa| |\mfA^{1/2} \mfb| = 
		|\mfA^{-1/2} \mfa| \bigl ( \mfb \cdot \mfA \mfb \bigr )^{1/2}. \label{eq:ACS}
	\end{align}

Now the function $(\dot \mfu_d, \dot \mfv_d, \dot \mfu, \dot \mfv)$ determined by \eqref{eq:mfmeq}, \eqref{eq:mfneq}, \eqref{eq:dotmfupeq} and \eqref{eq:dotmfvpeq} is an element of $L^2([0,L]; (\bbR^3)^4)$, satisfies the constraint \eqref{eq:const}, and 
\begin{align}
\cl F(\dot \msu_d, \dot \msv_d, \dot \msu, \dot \msv) &= \int_0^L \Bigl (
	(\mfm - \p_{\mfu} \psi) \cdot \mfM^{-1} (\mfm - \p_{\mfu} \psi) +
	(\mfn - \p_{\mfv} \psi) \cdot \mfN^{-1} (\mfn - \p_{\mfv} \psi) \\ &\qquad \qquad + \p_{\mfu_d} \psi \cdot \mfM_d^{-1} \p_{\mfu_d} \psi + \p_{\mfv_d} \psi \cdot \mfN_d^{-1} \p_{\mfv_d} \psi \Bigr ) ds \\
	&= \int_0^L \Bigl (
	|\mfM^{-1/2} (\mfm - \p_{\mfu} \psi)|^2 + 
	|\mfN^{-1/2} (\mfn - \p_{\mfv} \psi)|^2 \\
	&\qquad \qquad + |\mfM_d^{-1/2} \p_{\mfu_d} \psi|^2 + |\mfN_d^{-1/2} \p_{\mfv_d} \psi|^2 \Bigr ) ds.  \label{eq:totdiss1}
\end{align}

Let 
 $(\mfa_d, \mfb_d, \mfa, \mfb) \in L^2([0,L]; (\bbR^3)^4)$ satisfy the constraint
\begin{gather}
	\int_0^L \Bigl (\mfa \cdot \mfM \mfa + \mfb \cdot \mfN \mfb + 
	\mfa_d \cdot \mfM_d \mfa_d + \mfb_d \cdot \mfN_d \mfb_d \Bigr ) ds = \\
	\int_0^L \Bigl [(\mfm - \p_{\mfu} \psi) \cdot \mfa +  
	(\mfn - \p_{\mfv} \psi) \cdot \mfb
	\\ - \p_{\mfu_d} \psi \cdot \mfa_d
	- \p_{\mfv_d} \psi \cdot \mfb_d \Bigr ] ds. 
	\label{eq:constraint1} 
\end{gather}
By \eqref{eq:A1/2}, \eqref{eq:ACS}, the Cauchy-Schwarz inequality and \eqref{eq:totdiss1}, we have 
\begin{align}
	[\cl F(\msa_d, \msb_d, \msa, \msb) ]^2 
	&= \Bigl [ \int_0^L \Bigl (
	(\mfm - \p_{\mfu} \psi) \cdot \mfa +  
	(\mfn - \p_{\mfv} \psi) \cdot \mfb \\
	&\qquad \qquad - \p_{\mfu_d} \psi \cdot \mfa_d
	- \p_{\mfv_d} \psi \cdot \mfb_d
	\Bigr ) ds \Bigr ]^2 \\
	&\leq \int_0^L \Bigl (
	|\mfM^{-1/2} (\mfm - \p_{\mfu} \psi)|^2 + 
	|\mfN^{-1/2} (\mfn - \p_{\mfv} \psi)|^2 \\
	&\qquad \qquad + |\mfM_d^{-1/2}\p_{\mfu_d} \psi|^2 + |\mfN_d^{-1/2} \p_{\mfv_d} \psi|^2 \Bigr ) ds \\
	&\times \int_0^L \Bigl ( |\mfM^{1/2}  \mfa|^2 + |\mfN^{1/2} \mfb|^2  \\ 
	&\qquad \qquad + |\mfM_d^{1/2} \mfa_d|^2 +  |\mfN_d^{1/2} \mfb_d|^2 \Bigr ) ds \\
	&= \cl F(\dot \mfu_d, \dot \mfv_d, \dot \mfu, \dot \mfv)
	\cl F (\mfa_d, \mfb_d, \mfa, \mfb), \label{eq:max}
\end{align}
with equality if and only if there exists $\lambda \in \bbR$ such that 
\begin{align}
	(\mfa_d, \mfb_d, \mfa, \mfb) = \lambda (\dot \mfu_d, \dot \mfv_d, \dot \mfu, \dot \mfv). \label{eq:equalitylambda}
\end{align}
Thus, $(\dot \mfu_d, \dot \mfv_d, \dot \mfu, \dot \mfv)$ determined by \eqref{eq:mfmeq}, \eqref{eq:mfneq}, \eqref{eq:dotmfupeq} and \eqref{eq:dotmfvpeq} is a maximizer in $L^2([0,L]; (\bbR^3)^4)$ of the constrained functional $\cl F$, and by 
\eqref{eq:equalitylambda}, \eqref{eq:const}, and \eqref{eq:constraint1}, we conclude that it is the unique such maximizer. 
\end{proof}

By a minor variation of the previous proof, we obtain the following variant of Proposition \ref{p:maxdiss}. 

\begin{prop}\label{p:maxdissuniform}
	Let $t \in (-\infty, \infty)$ be fixed, and let
	\begin{gather}
		\mfu_d(t), \mfv_d(t) \in \bbR^3, \\
		\mfm(\cdot,t), \mfn(\cdot,t), \mfu(\cdot,t), \mfv(\cdot,t) \in L^\infty([0,L]; \bbR^3).
	\end{gather}
	The element 
	\begin{align}
		(\dot \mfu_d(t), \dot \mfv_d(t), \dot \mfu(\cdot,t), \dot \mfv(\cdot,t)) \in (\bbR^3)^2 \times L^2([0,L]; (\bbR^3)^2)
	\end{align}
	determined by the relations
	\begin{align}
		\mfm &= \p_{\mfu} \hat \psi(\mfu_d, \mfv_d, \mfu, \mfv) + \mfM \dot \mfu, \label{eq:mfmequnif} \\
		\mfn &= \p_{\mfv} \hat \psi(\mfu_d, \mfv_d, \mfu, \mfv) + \mfN \dot \mfv, \label{eq:mfnequnif} \\
		\mfM_d \dot \mfu_d &= -\int_0^L \p_{\mfu_d} \hat \psi(\mfu_d, \mfv_d, \mfu, \mfv) ds, \label{eq:dotmfupequnif} \\
		\mfN_d \dot \mfv_d &= -\int_0^L \p_{\mfv_d} \hat \psi(\mfu_d, \mfv_d, \mfu, \mfv) ds \label{eq:dotmfvpequnif}, 
	\end{align}
	is the unique maximizer in $(\bbR^3)^2 \times L^2([0,L]; (\bbR^3)^2)$ of the functional 
	\begin{align}
		\cl F(\dot \msu_d, \dot \msv_d, \dot \msu, \dot \msv) &= \int_0^L \Bigl ( \dot \mfu \cdot \mfM \dot \mfu 
		+ \dot \mfv \cdot \mfN \dot \mfv 
		+ \dot \mfu_d \cdot \mfM_d \dot \mfu_d + 
		\dot \mfv_d \cdot \mfN_d \dot \mfv_d \Bigr ) ds
	\end{align}
	subject to the constraint
	 \begin{gather}
	\int_0^L \Bigl (
	\dot \mfu \cdot \mfM \dot \mfu + \dot \mfv \cdot \mfN \dot \mfv + 
	\dot \mfu_d \cdot \mfM_d \dot \mfu_d + \dot \mfv_d \cdot \mfN_d \dot \mfv_d \Bigr ) ds
	= \\ \int_0^L \Bigl [ (\mfm - \p_{\mfu} \psi) \cdot \dot \mfu +  
	(\mfn - \p_{\mfv} \psi) \cdot \dot \mfv
	- \p_{\mfu_d} \psi \cdot \dot \mfu_d
	- \p_{\mfv_d} \psi \cdot \dot \mfv_d \Bigr ] ds.  
	\label{eq:constunif} 
\end{gather} 
For the maximizer, we have  
\begin{align}
	\cl F(\dot \msu_d, \dot \msv_d, \dot \msu, \dot \msv) &= \int_0^L 
	|\mfM^{-1/2} (\mfm - \p_{\mfu} \psi)|^2 + 
	|\mfN^{-1/2} (\mfn - \p_{\mfv} \psi)|^2 \, ds \\
	&\qquad + \Bigl |\mfM_d^{-1/2} \int_0^L \p_{\mfu_d} \psi \, ds \Bigr |^2 + \Bigl |\mfN_d^{-1/2} \int_0^L \p_{\mfv_d} \psi \, ds\Bigr |^2.  \label{eq:totdissipationunif} 
\end{align}
\end{prop} 

As advocated by Rajagopal and Srinivasa \cite{RAJAGOPALSRI2000}, we choose the forms of \eqref{eq:constequs} and \eqref{eq:evoleqs} based on the following requirement. 
\begin{assum}\label{assum:rajsri}
 At each fixed time $t$, and for each fixed set of vector fields defined at time $t$, 
$
\mfm, \mfn, \mfu_d, \mfv_d, \mfu, \mfu_d,$ and $\mfv_d$, 
the element $(\p_t \mfu_d, \p_t \mfv_d, \p_t \mfu, \p_t \mfv)$ maximizes the total dissipation rate subject to the constraint \eqref{eq:const1} at time $t$.	
\end{assum}
The only ambiguity in the above requirement is \emph{in what space} the element \\ $(\p_t \mfu_d, \p_t \mfv_d, \p_t \mfu, \p_t \mfv)$ maximizes the total dissipation rate, subject to the constraint \eqref{eq:const1}, at time $t$. The choice of space, $$\mbox{either } L^2([0,L]; (\bbR^3)^4) \mbox{ or } (\bbR^3)^2 \times L^2([0,L]; (\bbR^3)^2),$$  along with the previous two propositions then determine the constitutive relations for the model. 

By Proposition \ref{p:maxdiss} and Proposition \ref{p:maxdissuniform}, our constitutive relations for the contact couple and contact force in terms of the kinematic variables are uniquely given by 
\begin{gather}
	\mfm = \p_{\mfu} \psi + \mfM \p_t \mfu \iff 
	m_k = \p_{u_k} \psi + M_{k\ell} \p_t u_\ell, \quad k = 1, 2, 3, \label{eq:mconsteq} \\
		\mfn = \p_{\mfv} \psi + \mfN \p_t \mfv \iff 
	n_k = \p_{v_k} \psi + N_{k\ell} \p_t v_\ell, \quad k = 1, 2, 3. \label{eq:nconsteq}
\end{gather}

Assumption \ref{assum:rajsri} in the setting of $L^2([0,L]; (\bbR^3)^4)$ 
and Proposition \ref{p:maxdiss} imply that the evolution equations determining the natural configuration of the rod are
\begin{gather}
	\mfM_d \p_t \mfu_d = -\p_{\mfu_d} \psi \iff 
	M_{d, k\ell} \p_t u_{d, \ell} = - \p_{u_{d,k}} \psi, \quad k = 1, 2, 3, \label{eq:unatevol}\\
		\mfN_d \p_t \mfv_d = -\p_{\mfv_d} \psi \iff 
	N_{d, k\ell} \p_t v_{d, \ell} = - \p_{v_{d,k}} \psi, \quad k = 1, 2, 3.
	\label{eq:vnatevol}
\end{gather}
However, Assumption \ref{assum:rajsri} in the smaller space $(\bbR^3)^2 \times L^2([0,L]; (\bbR^3)^4)$ and Proposition \ref{p:maxdissuniform} imply that the evolution equations determining the natural configuration of the rod are  
\begin{gather}
	\mfM_d \p_t \mfu_d = -\int_0^L \p_{\mfu_d} \psi \, ds \iff \\ 
	M_{d, k\ell} \p_t u_{d, \ell} = - \int_0^L \p_{u_{d,k}} \psi \, ds, \quad k = 1, 2, 3, \label{eq:unatevolunif} \\
	\mfN_d \p_t \mfv_d = -\int_0^L \p_{\mfv_d} \psi \, ds \iff \\
	N_{d, k\ell} \p_t v_{d, \ell} = - \int_0^L \p_{v_{d,k}} \psi \, ds, \quad k = 1, 2, 3. \label{eq:vnatevolunif}
\end{gather}
The evolution equations \eqref{eq:unatevolunif} and \eqref{eq:vnatevolunif}, rather than \eqref{eq:unatevol} and \eqref{eq:vnatevol}, are consistent with the simplifying requirement that the natural configuration is homogeneous (independent of the variables $s$), at each time $t$, regardless of the current configuration. Ultimately, it is up to the modeler to adopt \eqref{eq:unatevol}, \eqref{eq:vnatevol} or \eqref{eq:unatevolunif}, \eqref{eq:vnatevolunif} as evolution equations for the natural configuration. 

\subsection{Dissipation rates}
We observe that for the constitutive relations \eqref{eq:mconsteq}, \eqref{eq:nconsteq}, \eqref{eq:unatevol}, \eqref{eq:vnatevol}, we have by \eqref{eq:disspw}, 
\begin{align}
	\xi = \p_t \msu \cdot \msM \p_t \msu + \p_t \msv \cdot \msN \p_t \msv + \p_t \msu_d \cdot \msM_d \p_t \msu_d + \p_t \msv_d \cdot \msN_d \p_t \msv_d \geq 0, 
\end{align} 
i.e., the point-wise dissipation rate is always non-negative. However, for the constitutive relations \eqref{eq:mconsteq}, \eqref{eq:nconsteq}, \eqref{eq:unatevolunif}, \eqref{eq:vnatevolunif}, we have
\begin{align}
	\xi = \p_t \msu \cdot \msM \p_t \msu + \p_t \msv \cdot \msN \p_t \msv + \p_{\msu_d} \psi \cdot \msM_d^{-1} \int_0^L \p_{\msu_d} \psi \, ds + \p_{\msv_d} \psi \cdot \msN_d^{-1} \int_0^L \p_{\msv_d} \psi \, ds, 
\end{align}
and $\int_0^L \xi \, ds \geq 0$, but $\xi$ is not necessarily point-wise non-negative. Indeed, let $s_0 \in [0,L]$. Suppose that $\msM_d$ is diagonal, and $\psi$ is the quadratic energy: 
\begin{align}
	\psi = \frac{1}{2} \sum_{k = 1}^3 \Bigl [
	A_{kk}(u_k - u_{d,k})^2 + B_{kk}(v_k - v_{d,k})^2
	+ A_{d,kk}u_{d,k}^2 + B_{d,kk}(v_{d,k} - \delta_{3k})^2
	\Bigr ],
\end{align}
with positive coefficients $A_{kk}$, $B_{kk}$, $A_{d,kk}$, $B_{d,kk}$, $k = 1,2,3$. Let $u \in C([0,L])$ with $u(s_0) > 0$ and $\int_0^L u(s) ds < 0$. Then with 
\begin{gather}
	\p_t \msu(s,t) = \p_t \msv(s,t) = \bs 0, \quad \msu_d(t) = \msv_d(t) - \bs e_3 = \bs 0, \\
	\msv(s,t) = \bs e_3, \quad \msu(s,t) = u(s) \bs e_3,  
\end{gather}
we have 
\begin{align}
	\xi(s_0,t) = \frac{A_{33}^2}{M_{33}} u(s_0) \int_0^L u(s) ds < 0.  
\end{align}

\subsection{Kelvin-Voight model} 
We observe that if the Helmholtz free energy is independent of the variables $\mfu_d$ and $\mfv_d$, $\psi = \hat \psi(\mfu, \mfv)$, then \eqref{eq:dotmfupeq} and \eqref{eq:dotmfvpeq} imply that the natural configuration is constant throughout the motion,
\begin{align}
	\mfu_d(\cdot, t) = \hat \mfu, \quad \mfv_d(\cdot, t) = \hat \mfv, \quad \forall t.
\end{align}
The constitutive equations for $\mfm$ and $\mfn$ are then given by classical Kelvin-Voight relations
\begin{align}
	\mfm = \p_{\mfu}\hat \psi(\mfu, \mfv) + \mfM \p_t \mfu, \\
	\mfn = \p_{\mfv}\hat \psi(\mfu, \mfv) + \mfN \p_t \mfv. 
\end{align} 
In particular, the limiting case $\mfM = \mfN = \bs 0$ yields the classical hyperelastic relations (see \cite{AntmanBook})
\begin{align}
	\mfm = \p_{\mfu}\hat \psi(\mfu, \mfv), \\
	\mfn = \p_{\mfv} \hat \psi(\mfu, \mfv). 
\end{align}

\subsection{Inextensible and unshearable rods}\label{s:inextunshconst}

The previous discussion was for extensible and shearable rods. We now briefly sketch here the simple changes needed to obtain constitutive relations if the rod is unshearable or if the rod is inextensible and unshearable. We will only give the sketch for general inhomogeneous evolution of the natural configuration, the changes for homogeneous evolution being clear.  

If the rod is unshearable, then we have the constraints 
\begin{align}
	v_{d,1} = v_{d,2} = v_{1} = v_2 = 0. \label{eq:unshconst}
\end{align}
The Helmholtz free energy and total dissipation rate then take the forms
\begin{align}
	\psi &= \hat \psi(\mfu_d, v_{d,3}, \mfu, v_3), \\
	\int_0^L \xi \, ds &= \int_0^L \Bigl (\p_t \mfu \cdot \mfM \p_t \mfu + M_{33} \p_t v_{3}^2 + 
	\p_t \mfu_d \cdot \mfM_d \p_t \mfu_d + M_{d,33} \p_t v_{d,3}^2 \Bigr ) ds,
\end{align}
and they satisfy the constraint
\begin{gather}
	\int_0^L (\p_t \mfu \cdot \mfM \p_t \mfu + M_{33} \p_t v_{3}^2 + 
	\p_t \mfu_d \cdot \mfM_d \p_t \mfu_d + M_{d,33} \p_t v_{d,3}^2 ) \, ds = \\
	\int_0^L \Bigl [ (\mfm - \p_{\mfu} \psi) \cdot \p_t \mfu +  
	(n_3 - \p_{v_3} \psi) v_3 
	- \p_{\mfu_d} \psi \cdot \p_t \mfu_d
	- (\p_{v_{d,3}} \psi) v_{d,3} \Bigr ] ds. 
\label{eq:constun} 
\end{gather}
Given $\mfm$, $\mfn$, $\mfu_d,$ $v_{d,3},$ $\mfu,$ and $v_3$, we maximize the total dissipation rate subject to the constraint \eqref{eq:constun} and arrive at the relations 
	\begin{align}
	\mfm &= \p_{\mfu} \hat \psi(\mfu_d, v_{d,3}, \mfu, v_3) + \mfM \p_t \mfu, \label{eq:mfmequ} \\
	n_3 &= \p_{v_3} \hat \psi(\mfu_d, v_{d,3}, \mfu, v_3) + M_{33} \p_t 
	v_{3}, \\
	\mfM_d \p_t \mfu_d &= -\p_{\mfu_d} \hat \psi(\mfu_d, v_{d,3}, \mfu, v_3), \label{eq:dotmfupequ} \\
	M_{d,33} \p_t v_{d,3} &= -\p_{v_{d,3}} \hat \psi(\mfu_d, v_{d,3}, \mfu, v_3). \label{eq:dotmfvpequ}
\end{align}
The components $n_1$ and $n_2$ are the indeterminate parts of $\mfn$ enforcing \eqref{eq:unshconst} during the motion of the rod. 

If the rod is inextensible and unshearable, then we have the constraints 
\begin{align}
	\mfv = \mfv_d = \bs e_3. \label{eq:inunshconst}
\end{align}
The Helmholtz free energy and dissipation rate then take the forms
\begin{align}
	\psi &= \hat \psi(\mfu_d, \mfu), \\
	\int_0^L \xi \, ds &= \int_0^L \Bigl (\p_t \mfu \cdot \mfM \p_t \mfu + 
	\p_t \mfu_d \cdot \mfM_d \p_t \mfu_d \Bigr ) ds,
\end{align}
and they satisfy the constraint
\begin{gather}
		\int_0^L ( \p_t \mfu \cdot \mfM \p_t \mfu + 
		\p_t \mfu_d \cdot \mfM_d \p_t \mfu_d ) ds
		= \\ \int_0^L \Bigl [ (\mfm - \p_{\mfu} \psi) \cdot \p_t \mfu
		- \p_{\mfu_d} \psi \cdot \p_t \mfu_d \Bigr ] ds. 
\label{eq:constinun} 
\end{gather}
Given $\mfm$, $\mfn$, $\mfu_d$ and $\mfu$, we maximize the total dissipation rate subject to the constraint \eqref{eq:constinun} and arrive at the relations 
\begin{align}
	\mfm &= \p_{\mfu} \hat \psi(\mfu_d, \mfu) + \mfM \p_t \mfu, \label{eq:mfmeqinu} \\
	\mfM_d \p_t \mfu_d &= -\p_{\mfu_d} \hat \psi(\mfu_d, \mfu). \label{eq:dotmfupeqinu}
\end{align}
The vector $\mfn$ is indeterminate and enforces \eqref{eq:inunshconst} during the motion of the rod. 

\section{Isolated torsion}

In the final section of this paper, we treat the problem of isolated torsion for a uniform rod with $(\rho A)(s)$, $I_{11}(s)$ and $I_{22}(s)$ fixed positive constants. We assume a quadratic free energy $\psi$ of the form, 
\begin{align}
\psi = \frac{1}{2} \sum_{k = 1}^3 \Bigl [
&A_{kk}(u_k - u_{d,k})^2 + B_{kk}(v_k - v_{d,k})^2 \\
&+ A_{d,kk}u_{d,k}^2 + B_{d,kk}(v_{d,k} - \delta_{3k})^2
\Bigr ], \label{eq:quadenergy}
\end{align} 
with positive coefficients $A_{kk}$, $B_{kk}$, $A_{d,kk}$, $B_{d,kk}$, $k = 1,2,3$. In addition, we assume that $\mfM, \mfM_d, \mfN$ and $\mfN_d$ are constant, positive definite, diagonal tensors.  We show that in contrast to the viscoelastic models considered in \cite{AntmanBook, AntmanSeidman05, LinnLangTuganov}, our model predicts both of the commonly observed phenomena of solid-like stress relaxation and creep.

\subsection{Dimensionless equations}

For isolated torsion the configuration satisfies 
\begin{gather}
	\mfu_d(s,t) = u_{d}(s,t) \bs e_3, \quad \mfu(s,t) = u(s,t) \bs e_3, \\
	\mfv(s,t) = \mfv_d(s,t) = \bs e_3, \quad \bs d_3(s,t) = \bs e_3, 
\end{gather}
and thus, 
\begin{gather}
	\p_s \bs r(s,t) = \bs d_3(s,t) = \bs e_3, \\
	\bs d_1(s,t) = \cos \phi(s,t) \bs e_1 + \sin \phi(s,t) \bs e_2, \quad 
	\bs d_2(s,t) = \bs d_3(s,t) \times \bs d_1(s,t),
\end{gather}
where $u_3 = u = \p_s \phi$. In what follows, we fix the orientation of the directors at the end $s = 0$ by requiring 
\begin{align}
	\phi(0,t) = 0 \iff \phi(s,t) = \int_0^s u(\sigma,t) d\sigma. 
\end{align}
See Figure \ref{fig:3}.  

\begin{figure}[t]
	\centering
	\includegraphics[width=\linewidth]{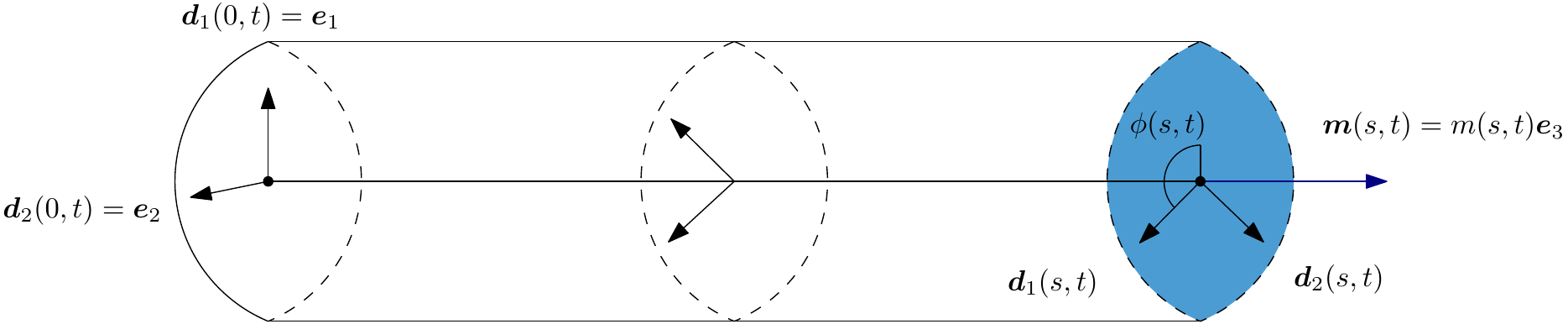}
	\caption{The set-up for isolated torsion.}
		\label{fig:3}
\end{figure}

Let $T$ be a characteristic time scale, and define dimensionless variables and quantities via
\begin{gather}
\bar s = \frac{s}{L}, \quad \bar t = \frac{t}{T}, \quad \bar \phi = \phi, \quad 
\bar u = L u, \quad \bar u_d = L u_d, \\
\bar{\hat \psi}(\bar u, \bar u_d) = \frac{T^2}{L^2 (\rho A)} \hat \psi(u_d \bs e_3, \bs e_3, u \bs e_3, \bs e_3), \\
\bar \mu_d = \frac{T M_{33}}{(\rho A) L^4}, \quad 
\bar \mu_d = \frac{T M_{d, 33}}{(\rho A) L^4}, \quad
\bar \nu = \frac{I_{II} + I_{22}}{(\rho A) L^2}. 
\end{gather}
By \eqref{eq:quadenergy}, the field equations \eqref{eq:linmom}, \eqref{eq:angmom} and the constitutive relations and evolution equations reduce to partial differential equations in dimensionless form (and dropping the over-bars)
\begin{align}
	\nu \p_t^2 \phi(s,t) &= \alpha [\p^2_s \phi(s,t) - \p_s u_d(s,t)] + \mu \p_{t} \p_s^2 \phi(s,t), \label{eq:phieq}\\ 
	\phi(0,t) &= 0, \\
	\mu_d \p_t u_d(s,t) &= \alpha [\p_s \phi(s,t) - u_d(s,t)] - \alpha_d u_d(s,t)\, \label{eq:upeq}
\end{align}
where $\mu, \mu_d, \nu, \alpha, \alpha_d \in (0,\infty)$. At the end $s = 1$, we prescribe the applied torsion $m(t)$ so that 
\begin{align}
	\alpha [\p_s \phi(1,t) - u_d(1,t)] + \mu \p_t \p_s \phi(1,t) = m(t). 
\end{align}
As discussed in Section 4, constraining the natural configuration to be homogeneous requires replacing \eqref{eq:upeq} with 
\begin{align}
	\mu_d \p_t u_d(t) = \int_0^1 \alpha [\p_s \phi(s,t) - u_d(t)] \, ds - \alpha u_d(t). \label{eq:upequniform}
\end{align}

The dimensionless constant $\nu$ can be interpreted as the ratio of the normalized second moment of area of the rod's cross section to the square of the length of the rod.\footnote{For a uniform rod of length $L$ and circular cross sections $A = \{(x_1, x_2) \mid x_1^2 + x_2^2 \leq \ell^2 \}$ of radius $\ell$, one can identify $$\nu = \frac{1}{L^2 (\pi \ell^2)} \int_{A} (x_1^2 + x_2^2) dx_1 dx_2 = \frac{1}{2} \frac{\ell^2}{L^2} \ll 1.$$ See Chapter 8 of \cite{AntmanBook}.} Since a rod is considered a slender body, it is therefore appropriate to assume that 
$
	\nu \ll 1.$ 
Quasi-static motion refers to the singular limiting case that $\nu = 0$, and the field equations reduce to \eqref{eq:upeq} (or \eqref{eq:upequniform}) and
\begin{align}
	\al(u(s,t) - u_d(s,t)) + \mu \p_t u(s,t) = m(t), \label{eq:quasistatic}
\end{align} 
where $m(t)$ is the applied torsion at the end $s = 1$ and $u = \p_s \phi$. If the current twist $u(t)$ and natural twist $u_d(t)$ are homogeneous in $s$, then $\phi(1,t) = u(t)$, and equations \eqref{eq:upeq} and \eqref{eq:upequniform} are identical.   

\subsection{Quasi-static motion: stress relaxation}

We now further consider quasi-static motion with homogeneous twist $u(t)$ and natural twist $u_d(t)$, and we show that the model predicts the two commonly observed viscoelastic phenomena of solid-like stress relaxation and creep.
In this setting, the field equations are 
\begin{gather}
	\mu \dot u(t) + \al u(t) - \al u_d(t) = m(t), \label{eq:twist}\\
	\mu_d \dot u_d(t) - \al u(t) + (\al + \al_d) u_d(t) = 0,  \label{eq:nattwist}
\end{gather}
where $\dot{\mbox{}} := \frac{d}{dt}$. 

We first consider the problem of determining the torsion $m(t)$ corresponding to a prescribed continuously differentiable twist history $\{ u(\tau) \mid \tau \in (-\infty, \infty) \}$ satisfying $u(\tau) = 0$ for all $\tau \in (-\infty, 0)$. We will also assume that the natural twist satisfies $u_d(\tau) = 0$ for all $\tau \in (-\infty,0)$.  Then \eqref{eq:nattwist} implies that 
\begin{align}
	u_d(t) = \frac{\al}{\mu_d} \int_0^t e^{-\frac{\al + \al_d}{\mu_d}(t - \tau)} u(\tau) d\tau, \quad t \in (-\infty, \infty),
\end{align}
whence 
\begin{align}
	m(t) = \mu \dot u(t) + \al u(t) - \frac{\al^2}{\mu_d} \int_0^t e^{-\frac{\al + \al_d}{\mu_d}(t - \tau)} u(\tau) d\tau, \quad t \in (-\infty, \infty). \label{eq:apptorsion}
\end{align}

Let 
\begin{align}
	H(\tau) = \begin{cases}
		0 &\mbox{ if } \tau \in (-\infty,0), \\
		1 &\mbox{ if } \tau \in [0,\infty),
	\end{cases}
\end{align}
be the Heaviside function. The applied torsion corresponding to the idealized step twist history $u(\tau) = u_0 H(t)$ 
is defined to be the limit, in the sense of distributions on $\bbR$, of the sequence of applied torsions $\{ m_n \}_n$ corresponding to a sequence of continuously differentiable, non-decreasing twist histories $\{ u_n \}_n$ satisfying 
\begin{align}
	u_n(\tau) = \begin{cases}
		0 &\mbox{ if } \tau \in (-\infty, 0), \\
		u_0 &\mbox{ if } \tau \in (1/n, \infty). 
	\end{cases}
\end{align}
We conclude that for $t \in [0,\infty)$,
\begin{align}
	u_d(t) &= \frac{\al}{\al + \al_d} u_0 \Bigl (1 - e^{-\frac{\al + \al_d}{\mu_d}t} \Bigr ), \\
	m(t) &= \mu u_0 \delta(t) + \frac{\al \al_d}{\al + \al_d} u_0 + \frac{\al^2}{\al + \al_d} u_0 e^{-\frac{\al + \al_d}{\mu_d}t}. \label{eq:stressrelax}
\end{align}
In particular, the asymptotic natural twist is given by 
\begin{align}
	u_d(\infty) = \frac{\al}{\al + \al_d}u_0,  
\end{align}
and for $u_0 > 0$, $m(t)$ is monotone decreasing on $(0,\infty)$ from $m(0^+) = \al u_0$ to 
$$m(\infty) = \frac{\al \al_d}{\al + \al_d}u_0.$$ Thus, the rod displays a form of \emph{solid-like stress relaxation}, a common viscoelastic phenomenon unaccounted for by the viscoelastic Cosserat rod models considered in \cite{AntmanBook, AntmanSeidman05, LinnLangTuganov}. We remark that it is clear from \eqref{eq:stressrelax} that stress relaxation persists in the case that dissipation is solely due to the evolution of the natural configuration, i.e., $\mu = 0$. 

\subsection{Quasi-static motion: creep}

Now we determine the twist $u(t)$ and natural twist $u_d(t)$ corresponding to a prescribed continuously differentiable torsion history $\{ m(\tau) \mid \tau \in (-\infty, \infty) \}$ satisfying $m(\tau) = 0$ for all $\tau \in (-\infty, 0)$. We rewrite \eqref{eq:twist} and \eqref{eq:nattwist} as 
\begin{gather}
	\frac{d}{dt}
	\begin{pmatrix}
		u(t) \\
		u_d(t) 
	\end{pmatrix}
= -\bs A 
	\begin{pmatrix}
	u(t) \\
	u_d(t) 
\end{pmatrix}
+ 	\begin{pmatrix}
	m(t)/\mu\\
	0
\end{pmatrix}, \\
\bs A = 
\begin{pmatrix}
{\al}/{\mu} & -{\al}/{\mu} \\
-{\al}/{\mu_d} &{(\al + \al_d)}/{\mu_d}
\end{pmatrix}. 
\end{gather} 
Observe that $\det \bs A = (\al \al_d)/(\mu \mu_d) > 0$, so $\bs A$ is non-singular.
If we assume that $u(t) = u_d(t) = 0$ for $t < 0$, we have by Duhamel's formula  
\begin{align}
	\begin{pmatrix}
	u(t) \\
	u_d(t) 
	\end{pmatrix}
= 
\int_0^t e^{-(t-\tau)\bs A} \begin{pmatrix}
	m(\tau)/\mu \\
	0
\end{pmatrix} d\tau, \quad t \in (-\infty, \infty). 
\end{align}
Let 
\begin{align}
	\delta(\al_d, \mu_d, \al, \mu) &= 
	\Bigl |
	\det (\bs A - \frac{1}{2} (\tr \bs A) \bs I)
	\Bigr |^{1/2} \\
	&=
	\Bigl [
	\Bigl 
	(
	\frac{\al}{2\mu} - \frac{\al + \al_d}{2\mu_d}
	\Bigr )^2 + \frac{\al^2}{\mu \mu_d} \Bigr ]^{1/2}
\end{align}
The matrix exponential is given explicitly by  
\begin{gather}
	e^{-t \bs A} = e^{-\Bigl (
		\frac{\al}{2\mu} + \frac{\al + \al_d}{2\mu_d}
		\Bigr )t} \cosh \bigl [\delta(\al_d, \mu_d, \al, \mu) t \bigr ]
	\begin{pmatrix}
		1 & 0 \\
		0 & 1 
	\end{pmatrix} \\
-  e^{-\Bigl (
	\frac{\al}{2\mu} + \frac{\al + \al_d}{2\mu_d}
	\Bigr )t} \sinh \bigl [\delta(\al_d, \mu_d, \al, \mu) t \bigr ]
\begin{pmatrix}
	\frac{\al}{2\mu} - \frac{\al+\al_d}{2\mu_d} & -\frac{\al}{\mu} \\
	-\frac{\al}{\mu_d} & \frac{\al + \al_d}{2\mu_d} - \frac{\al}{2\mu}.
\end{pmatrix} 
\end{gather}
and since $\frac{\al}{2\mu} + \frac{\al+\al_d}{2\mu_d} > \delta(\al_d, \mu_d, \al, \mu)$, we have 
\begin{align}
	\lim_{t \rar \infty} e^{-t\bs A} \bs a = \bs 0, \quad \bs a \in \bbR^3. \label{eq:vanishingmatexp} 
\end{align}

The twist and natural twist corresponding to the idealized step torsion history $m(\tau) = m_0 H(t)$ 
is defined to be the limit, in the sense of distributions on $\bbR$, of the sequences of twists $\{ u_n \}$ and natural twists $\{ u_{d,n} \}_n$ corresponding to a sequence of continuously differentiable, non-decreasing torsion histories $\{ m_n \}_n$ satisfying 
\begin{align}
	m_n(\tau) = \begin{cases}
		0 &\mbox{ if } \tau \in (-\infty, 0), \\
		m_0 &\mbox{ if } \tau \in (1/n, \infty). 
	\end{cases}
\end{align}
We conclude that for $t \in [0,\infty)$,
\begin{align}
	\begin{pmatrix}
		u(t) \\
		u_d(t) 
	\end{pmatrix}
= \Bigl (
\bs I - e^{-t\bs A}
\Bigr ) \bs A^{-1} 
\begin{pmatrix}
	m_0/ \mu \\
	0
\end{pmatrix} \label{eq:creepform}
\end{align}
By \eqref{eq:creepform} and \eqref{eq:vanishingmatexp} we conclude that the twist and natural twist start at $u(0) = 0, u_d(0) = 0$ and asymptotically tend to 
\begin{align}
	u(\infty) = \frac{\al + \al_d}{\al \al_d} m_0, \quad
	u_d(\infty) = 
	\frac{1}{\al_d} m_0, \label{eq:nattwistasymp}
\end{align}
a form of {\emph{solid-like creep}. 
	
In the case that $\mu = 0$, the previous computations do not apply, but the calculations are simpler. Indeed, if $\mu = 0$, then \eqref{eq:twist} and \eqref{eq:nattwist} are equivalent to  
\begin{gather}
	\alpha (u(t) - u_d(t)) = m(t), \\
	\mu_d \dot u_d(t) + \al_d u_d(t) = m(t). 
\end{gather} 
Assuming that $u(t) = u_d(t) = 0$ for $t < 0$, we have 
\begin{align}
	u(t) &= m(t) + \frac{\al}{\mu_d} \int_0^t e^{{-\al_d}{\mu_d}(t - \tau)}m(\tau) d\tau, \\
	u_d(t) &= \frac{1}{\mu_d} \int_0^t e^{-\frac{\al_d}{\mu_d}(t-\tau)} m(\tau) d\tau.
\end{align}
In the case of an idealized step torsion history $m(\tau) = m_0 H(t)$, we conclude that for $t \in [0,\infty)$, 
\begin{align}
	u(t) &= \frac{1}{\al} m_0 + \frac{1}{\al_d}\Bigl ( 1 - e^{-\frac{\al_d}{\mu_d} t} \Bigr ) m_0, \\
	u_d(t) &= \frac{1}{\al_d}\Bigl ( 1 - e^{-\frac{\al_d}{\mu_d} t} \Bigr ) m_0,
\end{align} 
and the twist and natural twist are strictly increasing in time with the same asymptotic values given in  \eqref{eq:nattwistasymp}. 

\bibliographystyle{plain}
\bibliography{researchbibmech}
\bigskip

\centerline{\scshape K. R. Rajagopal}
\smallskip
{\footnotesize
	\centerline{Department of Mechanical Engineering, Texas A\&M University}
	
	\centerline{College Station, TX 77843, USA}
	
	\centerline{\email{krajagopal@tamu.edu}}
}

\bigskip

\centerline{\scshape C. Rodriguez}
\smallskip
{\footnotesize
	\centerline{Department of Mathematics, University of North Carolina}
	
	\centerline{Chapel Hill, NC 27599, USA}
	
	\centerline{\email{crodrig@email.unc.edu}}
}

\end{document}